\documentclass[reqno,a4paper]{amsart}

\usepackage{amssymb}

\newcounter{mnotecount}[section]

\usepackage{mathrsfs}
\usepackage{eucal}

\usepackage{ipa} 

\DeclareMathOperator{\tho}{\text{\rm \thorn}}
\DeclareMathOperator{\edt}{\text{\rm \eth}}

\usepackage{color}
\usepackage{graphicx}
\numberwithin{equation}{section}

\usepackage[linktocpage]{hyperref}

\allowdisplaybreaks[2]

\theoremstyle{plain}
\newtheorem{theorem}{Theorem}[section]
\newtheorem{proposition}[theorem]{Proposition}
\newtheorem{lemma}[theorem]{Lemma}
\newtheorem{corollary}[theorem]{Corollary}

\newtheorem{definition}[theorem]{Definition}
\newtheorem{remark}[theorem]{Remark}
\newtheorem{example}[theorem]{Example}

\renewcommand{\Re}{{\mathbb R}}         

\newcommand{\Co}{{\mathbb C}}         
\newcommand{\half}{\frac{1}{2}}         
\newcommand{\eps}{\epsilon}

\newcommand{\Mink}{\mathbb M}

\newcommand{\NatNum}{\mathbb N}
   
\newcommand{\Mcal}{\mathcal M}

\newcommand{\sDiv}{\mathscr{D}}
\newcommand{\sCurl}{\mathscr{C}}
\newcommand{\sCurlDagger}{\mathscr{C}^\dagger}
\newcommand{\sTwist}{\mathscr{T}}
\newcommand{\sExt}{\mathscr{E}}

\newcommand{\SL}{\mathrm{SL}}
\newcommand{\SO}{\mathrm{SO}}
\newcommand{\SU}{\mathrm{SU}}

\newcommand{\Spin}{\text{Spin}}

\newcommand{\KillSpin}{\mathcal{KS}}

\newcommand{\PetrovD}{D}
\newcommand{\PetrovN}{N}
\newcommand{\PetrovO}{O}

\newcommand{\geodenergy}{\mathbf{e}}
\newcommand{\geodmass}{\mathbf{\mu}}
\newcommand{\geodcarter}{\mathbf{k}}

\newcommand{\NPl}{l}
\newcommand{\NPn}{n}
\newcommand{\NPm}{m}
\newcommand{\NPmbar}{\bar m}

\newcommand{\KSigma}{\Sigma} 
\newcommand{\KDelta}{\Delta}

\newcommand{\Lie}{\mathcal L} 

\newcommand{\SpaceSlice}{\Sigma} 
\newcommand{\Spaceh}{h}
\newcommand{\SpaceK}{k}
\newcommand{\ordo}{o}
\newcommand{\met}{g} 

\newcommand{\ba}{\mathbf{a}}
\newcommand{\bb}{\mathbf{b}}
\newcommand{\bc}{\mathbf{c}}
\newcommand{\bA}{\mathbf{A}}
\newcommand{\bB}{\mathbf{B}}
\newcommand{\bC}{\mathbf{C}}

\newcommand{\SymSpin}{S}
\newcommand{\SymSpinSec}{\mathcal S}

\newcommand{\Dnormal}{\nabla_\tau} 
\newcommand{\DSen}{\nabla}

\newcommand{\GenVec}{\nu} 

\newcommand{\EMTensorT}{X}
\newcommand{\etavarphi}{\varsigma}

\newcommand{\DiffCurlyRTilde}{{\tilde{\mathcal{R}}'}{}}

\newcommand{\newcommandMG}[3]{\newcommand{#1}{#3}} 
\newcommandMG{\fnMna} {{f_{1}}} {z}
\newcommandMG{\fnMnb} {{f_{2}}} {w}
\newcommandMG{\fnMca} {f_{1,1}}   {z_1}
\newcommandMG{\fnMcb} {f_{2,1}}   {w_1}
\newcommandMG{\fnMda} {f_{1,2}}   {z_2}
\newcommandMG{\fnMdb} {f_{2,2}}   {w_2}

\newcommand{\rp}{r_+}

\newcommand{\dr}{\partial_r}

\newcommand{\dt}{\partial_t}
\newcommand{\dtheta}{\partial_\theta}
\newcommand{\dphi}{\partial_\phi}

\newcommand{\di}{\mathrm{d}} 

\newcommand{\Geodesic}{\gamma}

\newcommand{\GeodesicEnergy}{\boldsymbol{e}}
\newcommand{\GeodesicLz}{\boldsymbol{\ell_z}}
\newcommand{\GeodesicQ}{\boldsymbol{q}}
\newcommand{\TensorQ}{Q}

\newcommand{\GenEnergyGeodesic}[1]{e_{#1}}
\newcommand{\rorbit}{r_{o}} 
\newcommand{\vecMGeodesic}{A}
\newcommand{\fnMrGeodesic}{\mathcal{F}}
\newcommand{\fnMrMorawetz}{\mathcal{F}}
\newcommand{\fnMpGeodesic}{q_{\text{reduced}}}

\newcommand{\gMetric}{g}

\newcommand{\vecMprimary}{A}
\newcommand{\scalMprimary}{q}
\newcommand{\vecMsecondary}{B}
\newcommand{\scalMsecondary}{s}

\newcommand{\HypersurfaceGeneral}{\Sigma}

\newcommand{\vecfont}[1]{#1}

\newcommand{\vecX}{\vecfont{X}}

\makeatletter
  \def\moverlay{\mathpalette\mov@rlay}
  \def\mov@rlay#1#2{\leavevmode\vtop{%
     \baselineskip\z@skip \lineskiplimit-\maxdimen
     \ialign{\hfil$#1##$\hfil\cr#2\crcr}}}
\makeatother

\newcommand{\squareTME}{\moverlay{\square\cr {\scriptscriptstyle \mathrm T}}}

\title{Spin geometry and conservation laws in the Kerr spacetime}

\author[L. Andersson]{Lars Andersson} \email{laan@aei.mpg.de}
\address{Albert Einstein Institute, Am M\"uhlenberg 1, D-14476 Potsdam,
  Germany 
\and
Department of Mathematics, Royal Institute of Technology, SE-100 44 Stockholm, Sweden
}

\author[T. B\"ackdahl]{Thomas B\"ackdahl} \email{t.backdahl@ed.ac.uk}
\address{The School of Mathematics, University of Edinburgh, James Clerk Maxwell Building, 
Peter Guthrie Tait Road, Edinburgh
EH9 3FD, UK}

\author[P. Blue]{Pieter Blue} \email{P.Blue@ed.ac.uk}
\address{The School of Mathematics and the Maxwell Institute, University of Edinburgh, James Clerk Maxwell Building, 
Peter Guthrie Tait Road, Edinburgh
EH9 3FD,UK}

\date{April 8, 2015}

\usepackage{comment}
\excludecomment{TOOLONG}

\begin{document} 

 \begin{abstract}

In this paper we will review some facts, both classical and recent, concerning the geometry and analysis of the Kerr and related black hole spacetimes. This includes the analysis of test fields on these spacetimes. Central to our analysis is the existence of a valence $(2,0)$ Killing spinor, which we use to construct symmetry operators and conserved currents as well as a new energy momentum tensor for the Maxwell test fields on a class of spacetimes containing the Kerr spacetime. We then outline how this new energy momentum tensor can be used to obtain decay estimated for Maxwell test fields. An important motivation for this work is the black hole stability problem, where fields with non-zero spin present interesting new challenges. The main tool in the analysis is  the 2-spinor calculus, and for completeness we introduce its main features.

\end{abstract}

\maketitle

\tableofcontents

\section{Introduction} \label{sec:intro}
\begin{TOOLONG}  
One of the goals of Einstein in introducing the general theory of
relativity was to find an explanation for the anomalous precession, and this
was also one of the first tests he put his new theory to. By calculating 
the orbits of test particles in the gravitational
field of a central body, using an 
approximate solution of the Einstein equations, he was able to give 
satisfactory explanation to the above mentioned anomaly. 
\end{TOOLONG} 
In the same month as Einstein's theory appeared, Karl Schwarzschild published
an exact 
and explicit solution of the Einstein vacuum equations describing
the gravitational field of a spherical body at rest. 
In analyzing Schwarzschild's solution, one finds that if the central body
is sufficiently concentrated, light emitted from its surface cannot reach an
observer at infinity. This phenomenon led John Archibald 
Wheeler to coin the term 
\emph{black hole} for this
type of object. 

It would take until the late 1950's before the global structure of
the Schwarzschild solution was completely understood 
and until the early 1970's before the idea that black holes exist in nature 
became widely accepted in the astrophysical
community. 
The reasons
for this can be traced to increasing observational evidence for compact
objects including neutron stars and quasars as well as an increasing
theoretical understanding of black holes. 

One of the most important
developments on the theoretical side was the 
discovery in 1963 by
Roy Kerr \cite{kerr:1963PhRvL..11..237K} of a new
explicit family of asymptotically flat solutions of the vacuum Einstein equations describing 
a stationary, rotating black hole. The Kerr family of solutions has
only two parameters, mass and azimuthal angular momentum, and includes the
Schwarzschild solution as a special case in the limit of vanishing angular
momentum. 

Assuming some technical conditions, any stationary asymptotically flat, stationary black hole spacetime is expected to belong to the Kerr family, a fact which is known to hold in the real-analytic case. Further, the Kerr black hole is expected to be stable in the sense that a small perturbation of the Kerr space time settles down asymptotically to a member of the Kerr family. In order to establish the astrophysical relevance of the Kerr solution, it is vital to find rigorous proofs of both of these conjectures, and a great deal of work has been devoted to these and related problems. 

In general, the orbits of test particles in the spacetime
surrounding a rotating object will be chaotic. However, in 1968
Brandon Carter \cite{carter:1968PhRv..174.1559C} discovered that the Kerr spacetime admits a conserved quantity not present in general
rotating spacetimes, 
known as the Carter constant,  and showed that the 
geodesic equation in the Kerr
spacetime can be integrated. 
This has allowed a detailed analysis of the behaviour of light and matter near a Kerr black hole, which has contributed substantially to the acceptance of the Kerr black hole as a fundamental object in astrophysics. 

The presence of the Carter constant is a manifestation of the separability and integrability properties of the Kerr spacetime. As discovered by Teukolsky \cite{teukolsky:1972PhRvL..29.1114T, teukolsky:1973}, the equations for test fields on the Kerr spacetime, including the scalar wave equation, the Dirac-Weyl, Maxwell and linearized gravity, are governed by a wave equation which admits separation of variables. These properties of the Kerr spacetime are analogues of the separability properties of the St\"ackel potentials which have been studied since the 19th century in Newtonian physics, and which have important applications in astrophysics. 

The Carter constant was shown by Walker and Penrose \cite{walker:penrose:1970CMaPh..18..265W} to originate in a Killing tensor, a notion originating in the work of Killing in the 1890s, cf. \cite{killing:1892}, and their and later work by Carter and others showed that the closely related Killing spinors are at the foundation of many of the remarkable properties of Kerr and other spacetimes admitting such objects. 

Although we shall here focus on symmetries and conservation laws related to the integrability properties of the Kerr and related spacetimes, the black hole stability problem is a fundamental motivation for this work. See section \ref{sec:BHstab} below for further discussion.  

\subsection{The Kerr solution} 
The Kerr metric describes a family of stationary, axisymmetric, asymptotically flat vacuum 
spacetimes, parametrized by ADM mass $M$ and angular momentum per unit mass $a$. 
In Boyer-Lindquist coordinates $(t,r,\theta,\phi)$, the Kerr metric takes the form\footnote{Here we have given the form of the metric has signature $+---$, which is most convenient when working with spinors, and which we shall use in this paper.} 
\begin{align}
\met_{ab}={}&\frac{(\Delta - a^2 \sin^2\theta) dt_{a} dt_{b}}{\Sigma}
 -  \frac{\Sigma dr_{a} dr_{b}}{\Delta}
 -  \Sigma d\theta_{a} d\theta_{b}
 - \frac{\sin^2\theta \bigl((a^2 + r^2)^2 - a^2 \sin^2\theta \Delta\bigr) d\phi_{a} d\phi_{b}}{\Sigma}\nonumber\\
& + \frac{2 a \sin^2\theta (a^2 + r^2 -  \Delta) dt_{(a}d\phi_{b)}}{\Sigma},\label{eq:met}
\end{align}
where
$\Delta = a^2 - 2 M r + r^2$ and $\Sigma = a^2 \cos^2\theta + r^2$.
The volume form is 
$\Sigma \sin\theta dt \wedge dr \wedge d\theta \wedge d\phi $.
For $|a|\leq M$, the Kerr spacetime contains a black hole, with event horizon at $r=r_+ \equiv M + \sqrt{M^2 - a^2}$.  In the subextreme case $|a| < M$, the surface gravity $\kappa = (r_+ - M)/(r_+^2 +a^2)$ is non-zero and the event horizon is non-degenerate. See \cite{ONeill,poisson:toolkit} for background on the geometry of the Kerr spacetime, see also \cite{FrolovNovikov}.

The Kerr metric admits two Killing vector fields $\xi^a = (\partial_t)^a$ (stationary) and 
$
(\partial_\phi)^a$ (axial). Although the stationary Killing field $\xi^a$ is timelike near infinity, since $\met(\partial_t,\partial_t)\rightarrow 1$ as $r\rightarrow
\infty$, $\xi^a$ becomes spacelike for $r$ sufficiently small,
when $1-2M/\KSigma<0$. In the Schwarzschild case $a=0$, this occurs at
the event horizon $r=2M$. However, for a rotating Kerr black hole with
$0<|a|\leq M$, there is an ergoregion outside the
event horizon where $\partial_t$ is spacelike. 
In the ergoregion, null and
timelike geodesics can have negative energy. Physically, it is
expected this means energy can be extracted from a rotating Kerr black
hole via the Penrose process, see \cite{finster:etal:energyextraction:MR2486663} and references therein.

The Kerr spacetime is expected to be the unique stationary, vacuum, asymptotically flat spacetime containing a non-degenerate black hole, see \cite{mars:2000CQGra..17.3353M,2013arXiv1304.0487A} and references therein, and is further expected to be dynamically stable. In fact, the scenario used by Penrose \cite{penrose:1973NYASA.224..125P} to motivate the important conjecture now known as the Penrose inequality\footnote{The Riemannian case of the Penrose inequality has been proved by Huisken and Ilmanen \cite{huisken:ilmanen:MR1916951} and Bray \cite{bray:MR1908823}.} involves, together with the weak cosmic censorship conjecture, the idea that the maximal vacuum Cauchy development of generic asymptotically flat vacuum data is asymptotic to a  Kerr spacetime.

Although a proof of uniqueness of Kerr is known for the real analytic case, and substantial progress on the uniqueness problem without this assumption has been made, the general case is still open. Similarly, the problem of dynamical stability of the Kerr spacetime has motivated a great deal of classical work exploiting the separability of the geometric field equations on Kerr, see eg. \cite{chandrasekhar:MR1210321,finster:etal:MR2525736} and references therein. This work however did not lead to pointwise decay estimates let alone with rates as one expects are needed to deal with the full nonlinear stablity problem. During the last decade, there has therefore been an intense focus on proving such estimates and progress has been made on proving such estimates for the wave, Dirac-Weyl, and Maxwell test fields on the Kerr spacetime, see \cite{andersson:blue:0908.2265,andersson:blue:2013arXiv1310.2664A} and references therein.  
At present, such estimates are not known for the equations of linearized gravity on Kerr. 

\subsection{Special geometry} 
A key fact concerning the Kerr spacetime, is that in addition to possessing the two Killing symmetries corresponding to stationarity and axial symmetry, the Kerr spacetime is algebraically special, with two repeated principal null directions for the Weyl tensor, i.e. it is of Petrov type $\PetrovD$.
This fact is closely related to the existence of the fourth constant of the motion for geodesics, discovered by Carter, as well as symmetry operators and separability properties for field equations in the Kerr spacetime. 
Algebraically special spaces have been the subject of intense study in the Lorentzian case, see for example \cite{stephani:etal:2009esef.book.....S}.  Although the Petrov classification has been extended to the Riemannian case \cite{Karlhede:1986}, see also \cite{hacyan:1979PhLA...75...23H,Batista:2012}, it has not played such an important role there.

In the Riemannian case, the special geometries which  have been most widely studied are the spaces with special holonomy. This class contains many of the most important examples, such as the Calabi-Yau and $G_2$ spaces. However, the Kerr black hole spacetime, arguably one of the most important Lorentz geometries and a central object in the present paper, does not have special holonomy, as can be seen from the fact that it has type $\PetrovD$\footnote{This is true also for the Riemannian signature version of the Kerr geometry.}. 
An important consequence of the algebraically special nature of the Kerr spacetime is that it admits a Killing spinor (or more properly, spin-tensor) of valence $(2,0)$, see section \ref{sec:prel}. As will be explained below, this fact implies the existence of symmetry operators and conserved currents. These symmetries may be called \emph{hidden} in the sense that they cannot be represented in terms of the Killing vector fields of the Kerr spacetime.

Riemannian spaces with special holonomy are characterized by the existence of parallel spinors \cite{M.Y.Wang:1989:MR1029845} or Killing spinors \cite{Baer:MR1224089}, a fact which extends also to Lorentzian spaces with special holonomy. The existence of a parallel spinor in the Riemannian case implies stability \cite{dai:wang:wei:MR2178660} in the sense of non-negativity of the spectrum of the Lichnerowicz Laplacian, a fact which applies to Calabi-Yau as well as $G_2$ spaces.  This fact is very closely related to the representation of linearized perturbations of spaces with parallel spinors discussed in  \cite{M.Y.Wang:1991:MR1129331}.  
Issues of stability are considerably more subtle in the Lorentzian case.

In a Lorentzian 4-manifold, the Hodge star operator acting on 2-forms has eigenvalues $\pm i$, while in a Riemannian 4-manifold, it has eigenvalues $\pm 1$. Hence, in the Lorentzian case, a real 2-form corresponds to a complex anti-self dual 2-form, while in the Riemannian case, a real 2-form may be split into self dual and anti-self dual parts. For this reason, there is no counterpart in the Lorentzian case to spaces with self-dual Weyl tensor, which form an important class of Riemannian 4-manifolds, containing e.g.  $K3$-surfaces and Gibbons-Hawking metrics. 
The just mentioned properties of the Hodge star are also closely related to the fact that the spin group in four dimensions with Lorentz signature is $\SL(2,\Co)$, with spin representations $\Co^2$ and $\bar \Co^2$, while in Riemannian signature, the spin group is $\SU(2) \times \SU(2)$ which acts on $\Co^2 \times \Co^2$ with independent action in each factor. 

The correspondence between spinors and tensors provides a particularly powerful tool in dimension four. In Lorentzian signature, the tensor product $\Co^2 \otimes \bar \Co^2$ of the two inequivalent spinor representations is naturally identified with the complexified Minkowski space. 
A similar situation obtains in the four dimensional Riemannian case with respect to the tensor product of the spin spaces $\Co^2 \otimes \Co^2$. Systematically decomposing expressions into their irreducible components  gives an effective tool for investigating the conditions for the existence of symmetry operators for field equations and conserved currents on Lorentzian 4-dimensional spacetimes. 

The SymManipulator package \cite{Bae11a}, which has been developed by one of the authors (T.B.) for the Mathematica based symbolic differential geometry suite xAct \cite{xAct}, exploits in a systematic way the above mentioned decompositions for the case of Lorentzian signature, and allows one to carry out investigations which are not feasible to do by hand. This has allowed the authors in recent work \cite{ABB:symop:2014CQGra..31m5015A} to complete and simplify the classification of second order symmetry operators and conserved currents for the spin-$s$ field equations for spins $0, 1/2, 1$ in general spacetimes. 

\subsection{Black hole stability} \label{sec:BHstab} 
The Kerr spacetime is expected to be dynamically stable, in the sense that the maximal development of Cauchy data close to Kerr data tend asymptotically in the future to a member of the Kerr family. 
The Black Hole Stability problem is to prove the just mentioned stability statement. This is one of the most important open problems in general relativity and has been the subject of intense work for the last decades.

Much of the work motivated by the Black Hole Stability problem, in particular during the 21st century has been directed towards understanding model problems, in particular to prove boundedness and decay in time for test fields on the Kerr spacetime, as well as on spacetimes which are asymptotic to Kerr in a suitable sense. For the case of scalar fields, i.e. solutions of the wave equation on the Kerr spacetime, these problems are now well understood, see \cite{andersson:blue:0908.2265,TataruTohaneanu,dafermos:rodnianski:etal:2014arXiv1402.7034D}. 

The full non-linear stability problem, however, has some features which are not present in the case of scalar fields. The Einstein equations have gauge symmetry in the form of diffeomorphism invariance (general covariance) and hence it is necessary to extract a hyperbolic system, either by performing a gauge reduction, or by extending the Einstein system. In addition to the gauge ambiguity, there is what one may term the moduli degrees of freedom of Kerr black hole spacetimes. 
Restricting our considerations to a black hole at rest with respect to an observer near infinity, the moduli space is parametrized by the Kerr parameters $a,M$. 
As mass and angular momentum is lost by radiation through null infinity, in the expected scenario of a maximal Cauchy development asymptotic to a Kerr black hole, the ``final'' parameters cannot be calculated from the given Cauchy data without actually solving the full Cauchy problem. 

We have a similar, but simpler situation if we consider the black hole stability problem for axi-symmetric spacetimes. 
In this case, the angular momentum is given by a Komar charge integral which is conserved, and hence the angular momentum is known a priori from the Cauchy data. In the case of zero angular momentum, the final state therefore must be a Schwarzschild black hole, and hence (disregarding boosts, translations etc.) the moduli space is reduced to having only one parameter, $M$. Also in this case, energy is lost through radiation, and the mass of the final black hole state cannot be determined a priori. The stability of the Schwarzschild black hole for the Einstein-scalar field system in spherical symmetry has been proved by Christodoulou \cite{christodoulou:MR885564}. 

Infinitesimal variations of the moduli parameters correspond to solutions of the linearized Einstein equations on the Kerr background
which do not disperse and can hence be described as non-radiating modes. 
The linearized Einstein equation is the equation for a field of spin $2$ which is relevant in this context.   
In order to prove dispersion for the ``radiating'' part of the linearized gravitational field, it is therefore necessary to eliminate these modes. 

The same phenomenon is present already in the case of the spin-1 or Maxwell field equation. 
For the case of Kerr, the domain of outer communication is diffeomorphic to the exterior of a cylinder in $\Re^4$. 
It follows that a source-free Maxwell field on the Kerr background can carry electric and magnetic charges.
The charge integrals are conserved, and hence a Maxwell field with non-zero charge cannot disperse. Similarly, for linearized gravity, the linearized mass and angular momentum correspond to conserved charge integrals, and hence solutions of linearized gravity with nonvanishing such charges cannot disperse, see \cite{aksteiner:andersson:2013CQGra..30o5016A}. 

This means that for the Maxwell field and for linearized gravity, it is not possible to prove dispersive (Morawetz)  estimates except by using a method which eliminates those solutions which ``carry'' the non-radiating modes. One approach to this is to make use of the linearity of the equations and explicitly subtract a suitable non-radiating solution so that the remainder has zero charges and will disperse. Another, and perhaps more direct approach is to use a projection which eliminates the ``non-radiating'' modes. Both for Maxwell and for linearized gravity on the Kerr background, such a projection can be found. 

In addition to the just mentioned difficulties, which are due to the non-trivial geometry of black hole spacetimes, the quadratic nature of the non-linearity in the Einstein equation makes it necessary to exploit cancellations in order to prove non-linear stability. This played a central role in the proofs of the nonlinear stability of Minkowski space by Christodoulou and Klainerman \cite{1993gnsm.book.....C} and by Lindblad and Rodnianski \cite{2005CMaPh.256...43L}, and related ideas must be included in any successful approach to the black hole stability problem. 

\subsection*{Overview of this paper} 
In section \ref{sec:prel} we introduce some background for the analysis in this paper, including some material on spin geometry and  
section \ref{sec:specialgeom} contains some material on  algebraically special spacetimes and spacetimes with Killing spinors.  
In section \ref{sec:kerrspacetime} we discuss some aspects of the Kerr geometry, the main example of the phenomena and problems discussed in this paper, in more detail. A new characterization of Kerr from the point of view of Killing spinors, fitting the perspective of this paper, is given section \ref{sec:characterization}. 
Section \ref{sec:hidsym} collects some results on symmetry operators and conserved currents due to the authors. The discussion of symmetry operators follows the paper \cite{ABB:symop:2014CQGra..31m5015A} while the results on conserved currents is part of ongoing work. A complete treatment will appear in \cite{ABB:currents}. The ideas developed in section \ref{sec:hidsym} is applied to the Teukolsky system in section \ref{sec:teuk}, where a new conserved stress-energy tensor for the Maxwell field is given, which can be argued to be the stress-energy tensor appropriate for the compbined spin-1 Teukolsky, and Teukolsky-Starobinsky system. Part of the results presented here can be found in \cite{2014arXiv1412.2960A}. 
Finally in section \ref{sec:morawetz}, we indicate how the new stress-energy  tensor  can be used to prove dispersive (Morawetz type) estimates for the Maxwell field on the Schwarzschild spacetime. This new result is part of ongoing work aimed at proving dispersive estimates for the Maxwell and linearized gravity field on the Kerr spacetime.

\section{Preliminaries} \label{sec:prel} 

In this paper we will make use of the 2-spinor  formalism, as well as the closely related GHP formalism. 
A detailed introduction to this material is given by Penrose and Rindler in \cite{Penrose:1986fk}.  Following the conventions there, we use the abstract index notation with lower case latin letters $a,b,c,\dots$ for tensor indices, and unprimed and primed upper-case latin letters $A,B,C, \dots, A', B', C', \dots$ for spinor indices. Tetrad and dyad indices are boldface latin letters following the same scheme,  $\ba,\bb,\bc,\dots, \bA,\bB,\bC,\dots,\bA',\bB',\bC',\dots$. For coordinate indices we use greek letters $\alpha, \beta, \gamma, \dots$.

\subsection{Spinors on Minkowski space} 
Consider Minkowski space $\Mink$, i.e. $\Re^4$ with coordinates $(x^\alpha) = (t, x, y, z)$ and metric 
$$
\met_{\alpha\beta} dx^\alpha dx^\beta = dt^2 - dx^2 - dy^2 - dz^2.
$$
Define a complex null tetrad (i.e. frame) $(\met_{\ba}{}^a)_{\ba = 0,\cdots, 3} = (\NPl^a, \NPn^a, \NPm^a, \NPmbar^a) $, normalized so that $\NPl^a \NPn_a = 1$, $\NPm^a \NPmbar_a = -1$,  
so that 
\begin{equation}\label{eq:metNP} 
\met_{ab} = 2 (\NPl_{(a} \NPn_{b)} - \NPm_{(a} \NPmbar_{b)}), 
\end{equation} 
by 
\begin{align*} 
\NPl^a = \met_{\mathbf 0}{}^a &= \frac{1}{\sqrt{2}}((\partial_t)^a + (\partial_z)^a),
&\NPn^a &=   \met_{\mathbf 1}{}^a = \frac{1}{\sqrt{2}}((\partial_t)^a - (\partial_z)^a), \\ 
\NPm^a = \met_{\mathbf 2}{}^a &= \frac{1}{\sqrt{2}} ( (\partial_x)^a - i (\partial_y)^a), 
&\NPmbar^a &=  \met_{\mathbf 3}{}^a = \frac{1}{\sqrt{2}} ( (\partial_x)^a + i (\partial_y)^a).
\end{align*} 
Similarly, let $\eps_\bA{}^A$ be a dyad (i.e. frame) in $\Co^2$, with dual frame $\eps_A{}^\bA$.
The complex conjugates will be denoted $\bar\eps_{\bA'}{}^{A'}, \bar\eps_{A'}{}^{\bA'}$ and again form a basis in another 2-dimensional complex space denoted $\bar \Co^2$, and its dual.  We can identify  the space of complex $2\times2$ matrices with $\Co^2 \otimes \bar \Co^2$. By construction, the tensor products $\eps_\bA{}^A \bar\eps_{\bA'}{}^{A'}$ and $\eps_A{}^\bA \bar\eps_{A'}{}^{\bA'}$forms a basis in $\Co^2 \otimes \bar \Co^2$ and its dual.  

Now, with $x^\ba = x^a \met_a{}^\ba $,  writing 
\begin{equation}\label{eq:IvW} 
x^\ba \met_\ba{}^{\bA \bA'} \equiv \begin{pmatrix} x^0 & x^2 \\ x^3 & x^1 \end{pmatrix} 
\end{equation} 
defines the soldering forms, also known as Infeld-van der Waerden symbols $\met_a{}^{AA'}$, (and analogously $\met_{AA'}{}^a$). 
By a slight abuse of notation we may write $x^{AA'} = x^a$ instead of $x^{\bA\bA'} = x^\ba \met_\ba{}^{\bA\bA'}$ or, dropping reference to the tetrad, $x^{AA'} = x^a \met_a{}^{AA'}$. 
In particular, we have that $x^a \in \Mink$ corresponds to a $2\times 2$ complex Hermitian 
matrix $x^{\bA\bA'} \in \Co^2 \otimes \bar \Co^2$.
Taking the complex conjugate of both sides of \eqref{eq:IvW} gives 
$$
\bar x^a = \bar x^{A'A} = (x^{AA'})^* .
$$
where $*$ denotes Hermitian conjugation. This extends to a correspondence $\Co^4 \leftrightarrow \Co^2 \otimes \bar \Co^2$ with complex conjugation corresponding to Hermitian conjugation. 

Note that 
\begin{equation}\label{eq:detmet}
\det(x^{\bA\bA'}) = x^0 x^1 - x^2 x^3 = x^a x_a /2.
\end{equation} 

We see from the above that the group 
$$
\SL(2,\Co) = \Bigl\{ A = \begin{pmatrix} a & b \\ c & d \end{pmatrix}, \quad a,b,c,d \in \Co, \quad ad-bc = 1 \Bigr\}
$$
acts on $X \in \Co^2 \otimes \bar \Co^2$ by 
$$
X \mapsto A X A^* .
$$
In view of \eqref{eq:detmet} this exhibits $\SL_2(\Co)$ as a double cover of the identity component of the Lorentz group $\SO_0(1,3)$, the group of linear isometries of $\Mink$. 
In particular, $\SL(2,\Co)$ is the spin group of $\Mink$. The canonical action
$$
(A, v) \in \SL(2,\Co) \times \Co^2 \mapsto A v \in \Co^2
$$
of $\SL(2,\Co)$ on $\Co^2$ is the spinor representation. Elements of $\Co^2$ are called (Weyl) spinors. The conjugate representation given by 
$$
(A, v) \in \SL(2,\Co) \times \Co^2 \mapsto \bar A v \in \Co^2
$$
is denoted $\bar\Co^2$.

Spinors\footnote{It is conventional to refer to spin-tensors eg. of the form $x^{AA'}$ or $\psi_{ABA'}$ simply as spinors.} 
of the form $x^{AA'} = \alpha^A \beta^{A'}$ correspond to matrices of rank one, and hence to complex null vectors. Denoting $o^A = \eps_0{}^A, \iota^A = \eps_1{}^A$, we have from the above that 
\begin{equation}\label{eq:tetrad-dyad} 
\NPl^a = o^A o^{A'}, \quad \NPn^a = \iota^A \iota^{A'}, \quad \NPm^a = o^A \iota^{A'}, \quad \NPmbar^a = \iota^A o^{A'}
\end{equation} 
This gives a correspondence between a null frame in $\Mink$ and a dyad in $\Co^2$. 

The action of $\SL(2,\Co)$ on $\Co^2$ leaves invariant a complex area element, a skew-symmetric bispinor. A unique such spinor $\eps_{AB}$ is determined by the normalization 
$$
\met_{ab} = \eps_{AB} \bar\eps_{A'B'}.
$$
The inverse $\eps^{AB}$ of $\eps_{AB}$ is defined by $\eps_{AB}\eps^{CB} = \delta_A{}^C$, $\eps^{AB} \eps_{AC} = \delta_C{}^B$.
As with $\met_{ab}$ and its inverse $\met^{ab}$, the spin-metric $\eps_{AB}$ and its inverse $\eps^{AB}$ 
is used to lower and raise spinor indices, 
$$
\lambda_{B} = \lambda^{A}\eps_{AB} , \quad \lambda^{A} = \eps^{AB} \lambda_{B}.
$$
We have
$$
\eps_{AB} = o_A \iota_B - \iota_A o_B.
$$
In particular,
\begin{equation} \label{eq:dyad-normalization} 
o_A \iota^A = 1.
\end{equation}

An element $\phi_{A\cdots D A' \cdots D'}$ of $\bigotimes^k \Co^2 \bigotimes^l \bar \Co^2$ is called a spinor of valence $(k,l)$. The space of totally symmetric\footnote{The ordering between primed and unprimed indices is irrelevent.} spinors $\phi_{A \cdots D A' \cdots D'} = \phi_{(A \cdots D) (A' \cdots D')}$ is denoted $\SymSpin_{k,l}$. The spaces $\SymSpin_{k,l}$ for $k,l$ non-negative integers yield all irreducible representations of $\SL(2,\Co)$. In fact, one can decompose any spinor into ``irreducible pieces'', i.e. as a linear combination of totally symmetric spinors in $\SymSpin_{k,l}$ with factors of $\eps_{AB}$. The above mentioned correspondence between vectors and spinors extends to tensors of any type, and hence the just mentioned decomposition of spinors into irreducible pieces carries over to tensors as well. 
Examples are given by $\mathcal F_{ab} = \phi_{AB} \eps_{A'B'}$, a complex anti-self-dual 2-form, and ${}^-C_{abcd} = \Psi_{ABCD} \eps_{A'B'} \eps_{C'D'}$, a complex anti-self-dual tensor with the symmetries of the Weyl tensor. Here, $ \phi_{AB}$ and $\Psi_{ABCD}$ are symmetric.

\subsection{Spinors on spacetime} 
Let now $(\Mcal, \met_{ab})$ be a Lorentizian 3+1 dimensional spin manifold with metric of signature $+---$. The spacetimes we are interested in here are spin, in particular any orientable, globally hyperbolic 3+1 dimensional spacetime is spin, cf. \cite[page 346]{Ger70spinstructII}.
If $\Mcal$ is spin, then the orthonormal frame bundle $\SO(\Mcal)$ admits a lift to $\Spin(\Mcal)$, a principal $\SL(2,\Co)$-bundle. The associated bundle construction now gives vector bundles over $\Mcal$ corresponding to the  representations of $\SL(2,\Co)$, in particular we have bundles of valence $(k,l)$ spinors with sections $\phi_{A\cdots D A' \cdots D'}$.
The Levi-Civita connection lifts to act on sections of the spinor bundles, 
\begin{equation}\label{eq:nablavarphi}
\nabla_{AA'} : \varphi_{B \cdots D B' \cdots D' } \to \nabla_{AA'} \varphi_{B \cdots D B' \cdots D'} 
\end{equation} 
where we have used the tensor-spinor correspondence to replace the index $a$ by $AA'$. We shall denote the totally symmetric spinor bundles by $\SymSpin_{k,l}$ and their spaces of sections by $\SymSpinSec_{k,l}$. 

The above mentioned correspondence between spinors and tensors, and the decomposition into irreducible pieces, can be applied to the Riemann curvature tensor. In this case, the irreducible pieces correspond to the scalar curvature, traceless Ricci tensor, and the Weyl tensor, denoted by  
$R$, $S_{ab}$, and $C_{abcd}$, respectively. The Riemann tensor then takes the form  
\begin{align}
R_{abcd}={}&- \tfrac{1}{12} g_{ad} g_{bc} R
 + \tfrac{1}{12} g_{ac} g_{bd} R
 + \tfrac{1}{2} g_{bd} S_{ac}
 -  \tfrac{1}{2} g_{bc} S_{ad}
 -  \tfrac{1}{2} g_{ad} S_{bc}
 + \tfrac{1}{2} g_{ac} S_{bd}
 + C_{abcd}.
\end{align}
The spinor equivalents of these tensors are
\begin{subequations}
\begin{align}
C_{abcd}={}&\Psi_{ABCD} \bar\epsilon_{A'B'} \bar\epsilon_{C'D'}+\bar\Psi_{A'B'C'D'} \epsilon_{AB} \epsilon_{CD},\\
S_{ab} ={}& -2 \Phi_{ABA'B'},\\
R={}&24 \Lambda.
\end{align}
\end{subequations}

Projecting \eqref{eq:nablavarphi} on its irreducible pieces gives the following four \emph{fundamental operators}.  
\begin{definition}
The differential operators
$$
\sDiv_{k,l}:\mathcal{S}_{k,l}\rightarrow \mathcal{S}_{k-1,l-1}, \quad 
\sCurl_{k,l}:\mathcal{S}_{k,l}\rightarrow \mathcal{S}_{k+1,l-1}, \quad 
\sCurlDagger_{k,l}:\mathcal{S}_{k,l}\rightarrow \mathcal{S}_{k-1,l+1}, \quad 
\sTwist_{k,l}:\mathcal{S}_{k,l}\rightarrow \mathcal{S}_{k+1,l+1}
$$
are defined as
\begin{subequations}
\begin{align}
(\sDiv_{k,l}\varphi)_{A_1\dots A_{k-1}}{}^{A_1'\dots A_{l-1}'}\equiv{}&
\nabla^{BB'}\varphi_{A_1\dots A_{k-1}B}{}^{A_1'\dots A_{l-1}'}{}_{B'},\\
(\sCurl_{k,l}\varphi)_{A_1\dots A_{k+1}}{}^{A_1'\dots A_{l-1}'}\equiv{}&
\nabla_{(A_1}{}^{B'}\varphi_{A_2\dots A_{k+1})}{}^{A_1'\dots A_{l-1}'}{}_{B'},\\
(\sCurlDagger_{k,l}\varphi)_{A_1\dots A_{k-1}}{}^{A_1'\dots A_{l+1}'}\equiv{}&
\nabla^{B(A_1'}\varphi_{A_1\dots A_{k-1}B}{}^{A_2'\dots A_{l+1}')},\\
(\sTwist_{k,l}\varphi)_{A_1\dots A_{k+1}}{}^{A_1'\dots A_{l+1}'}\equiv{}&
\nabla_{(A_1}{}^{(A_1'}\varphi_{A_2\dots A_{k+1})}{}^{A_2'\dots A_{l+1}')}.
\end{align}
\end{subequations}
The operators are called respectively the divergence, curl, curl-dagger, and twistor operators. 
\end{definition}
With respect to complex conjugation, the operators $\sDiv, \sTwist$ satisfy $\overline{\sDiv_{k,l}} = \sDiv_{l,k}$, $\overline{\sTwist_{k,l}} = \sTwist_{l,k}$, while $\overline{\sCurl_{k,l}} = \sCurlDagger_{l,k}$, $\overline{\sCurlDagger_{k,l}} = \sCurl_{l,k}$. 

Denoting the adjoint of an operator by $\mathcal A$ with respect to the bilinear pairing 
$$
(\phi_{A_1 \cdots A_k A'_1 \cdots A'_l}, \psi_{A_1 \cdots A_k A'_1 \cdots A'_l})=\int \phi_{A_1 \cdots A_k A'_1 \cdots A'_l} \psi^{A_1 \cdots A_k A'_1 \cdots A'_l}d\mu 
$$
by $\mathcal A^\dagger$, 
and the adjoint with respect to the sesquilinear pairing  
$$
\langle \phi_{A_1 \cdots A_k A'_1 \cdots A'_l}, \psi_{A_1 \cdots A_l A'_1 \cdots A'_k}\rangle =\int \phi_{A_1 \cdots A_k A'_1 \cdots A'_l} \bar\psi^{A_1 \cdots A_k A'_1 \cdots A'_l}d\mu 
$$
by $\mathcal A^\star$ , 
we have 
\begin{align*}
(\sDiv_{k,l})^\dagger &= - \sTwist_{k-1,l-1}, & (\sTwist_{k,l})^\dagger &= - \sDiv_{k+1,l+1},&
(\sCurl_{k,l})^\dagger &= \sCurlDagger_{k+1,l-1}, & (\sCurlDagger_{k,l})^\dagger &= \sCurl_{k-1,l+1},
\end{align*}
and
\begin{align*}
(\sDiv_{k,l})^\star &= - \sTwist_{l-1,k-1} , & (\sTwist_{k,l})^\star &= - \sDiv_{l+1,k+1},&
(\sCurl_{k,l})^\star &= \sCurl_{l-1,k+1}, & (\sCurlDagger_{k,l})^\star &= \sCurlDagger_{l+1,k-1}.
\end{align*}

As we will see in section~\ref{sec:masslessspins}, the kernels of $\sCurlDagger_{2s,0}$ and $\sCurl_{0,2s}$ are the massless spin-s fields. The kernels of $\sTwist_{k,l}$, are the valence $(k,l)$ Killing spinors, which we will discuss further in section~\ref{sec:KillingSpinors} and section~\ref{sec:KillSpinSpacetime}.
A multitude of commutator properties of these operators can be found in \cite{ABB:symop:2014CQGra..31m5015A}.

\subsection{GHP formalism} 
Given  
a null tetrad $\NPl^a, \NPn^a, \NPm^a, \NPmbar^a$ we have a spin dyad $o_A, \iota_A$ as discussed above. For a spinor $\varphi_{A\cdots D} \in \SymSpinSec_{k,0}$, it is convenient to introduce the Newman-Penrose  scalars 
\begin{equation}\label{eq:varphi_i-def}
\varphi_i = \varphi_{A_1 \cdots A_i A_{i+1} \cdots A_k} \iota^{A_1} \cdots \iota^{A_i} o^{A_{i+1}} \cdots o^{A_k}.
\end{equation} 
In particular, $\Psi_{ABCD}$ corresponds to the five complex Weyl scalars 
$\Psi_i, i = 0, \dots 4$.  The definition $\varphi_i$ extends in a natural way to the scalar components of spinors of valence $(k,l)$.

The normalization \eqref{eq:dyad-normalization} is left invariant under rescalings $o_A \to \lambda o_A$, $\iota_A \to \lambda^{-1} \iota_A$ where $\lambda$ is a non-vanishing complex scalar field on $\Mcal$. Under such rescalings, the scalars defined by projecting on the dyad, such as $\varphi_i$ given by \eqref{eq:varphi_i-def} transform as sections of complex line bundles. A scalar $\varphi$ is said to have type $\{p,q\}$ if $\varphi \to \lambda^p \bar\lambda^q \varphi$ under such a rescaling. Such fields are called properly weighted. The lift of the Levi-Civita connection $\nabla_{AA'}$ to these bundles gives a covariant derivative denoted $\Theta_a$. Projecting on the null tetrad $\NPl^a, \NPn^a, \NPm^a, \NPmbar^a$ gives the GHP operators 
$$
\tho = \NPl^a \Theta_a, \quad \tho' = \NPn^a\Theta_a, \quad \edt = \NPm^a \Theta_a , \quad 
\edt' = \NPmbar^a \Theta_a.
$$
The GHP operators are properly weighted, in the sense that they take properly weighted fields to properly weighted fields, for example if $\varphi$ has type $\{p,q\}$, then $\tho \varphi$ has type $\{p+1, q+1\}$. This can be seen from the fact that $\NPl^a = o^A \bar{o}^{A'}$ has type $\{1,1\}$. 
There are 12 connection coefficients in a null frame, up to complex conjugation. Of these, 8 are properly weighted, the GHP spin coefficients. The other connection coefficients enter in the connection 1-form for the connection $\Theta_a$. 

The following formal operations take weighted quantities to weighted
quantities,
\begin{equation}\label{eq:ghpsym}
\begin{aligned}
^-(\text{bar})&: \; \NPl^a \to \NPl^a, \;   \NPn^a \to \NPn^a, \;  \NPm^a
\to \NPmbar^a,	\; \NPmbar^a \to \NPm^a, &\{p,q\}\to\{q,p\} ,\\
'(\text{prime})&: \; \NPl^a \to \NPn^a, \;   \NPn^a \to \NPl^a, \;
\NPm^a \to \NPmbar^a,  \; \NPmbar^a \to \NPm^a, &\{p,q\}\to\{-p,-q\} ,\\
^*(\text{star})&: \; \NPl^a \to \NPm^a , \;  \NPn^a \to -\NPmbar^a, \;
\NPm^a \to -\NPl^a, \; \NPmbar^a \to \NPn^a, &\{p,q\}\to\{p,-q\} .
\end{aligned}
\end{equation}
The properly weighted spin coefficients can be represented as 
\begin{align}
\kappa &=  \NPm^b \NPl^a \nabla_a \NPl_b , \quad
\sigma =  \NPm^b \NPm^a \nabla_a \NPl_b, \quad
\rho   =  \NPm^b \NPmbar^a \nabla_a \NPl_b , \quad
\tau   = \NPm^b \NPn^a \nabla_a \NPl_b ,
\label{eq:spincoeff-def}
\end{align}
together with their primes $\kappa', \sigma', \rho', \tau'$.

A systematic application of the above formalism allows one to write the tetrad projection of the geometric field equations in a compact form. For example, the Maxwell equation corresponds to the four scalar equations given by 
\begin{align}
(\tho -2\rho)\phi_1 -({\edt}'-\tau')\phi_0 = -\kappa \phi_2 , \label{eq:ghpmaxwell}
\end{align}
with its primed and starred versions.

Working in a spacetime of Petrov type $\PetrovD$ gives drastic simplifications, in view of the fact that 
choosing the null tedrad so that $\NPl^a$, $\NPn^a$ are aligned with principal null directions of the Weyl tensor (or equivalently choosing the spin dyad so that $o_A, \iota_A$ are principal spinors of the Weyl spinor), as has already been mentioned, the Weyl scalars are zero with the exception of $\Psi_2$, and the only non-zero spin coefficients are $\rho, \tau$ and their primed versions.

\subsubsection{Massless spin-$s$ fields} \label{sec:masslessspins}
For $s \in \half \NatNum$,  $\varphi_{A\cdots D} \in \ker  \sCurlDagger_{2s,0}$ is a totally symmetric spinor $\varphi_{A\cdots D} = \varphi_{(A \cdots D)}$ of valence $(2s,0)$ which solves the massless spin-s equation 
$$
(\sCurlDagger_{2s,0} \varphi)_{A'B\cdots D} = 0.
$$
For $s=1/2$, this is the Dirac-Weyl equation $\nabla_{A'}{}^A \varphi_A = 0$, for $s=1$, we have the left and right Maxwell equation $\nabla_{A'}{}^B \phi_{AB} = 0$ and $\nabla_A{}^{B'} \varphi_{A'B'} = 0$, i.e. $(\sCurlDagger_{2,0} \phi)_{A'A} =0$, $(\sCurl_{0,2} \varphi)_{AA'} = 0$. 

An important example is the Coulomb Maxwell field on Kerr,   
\begin{equation}\label{eq:coulomb} 
\phi_{AB} = - \frac{2}{(r-ia\cos\theta)^2} o_{(A} \iota_{B)}
\end{equation} 
This is a non-trivial sourceless solution of the Maxwell equation on the Kerr background.
We note that $\phi_1 = (r - i a \cos\theta)^{-2}$ while $\phi_0 = \phi_2 = 0$. 

For $s > 1$, the existence of a non-trivial solution to the spin-s equation implies  curvature conditions, a fact known as the Buchdahl constraint. \cite{Buchdahl58}
\begin{equation}
0=\Psi_{(A}{}^{DEF}\phi_{B \dots C)DEF}.
\end{equation}
This is easily obtained by commuting the operators in
\begin{equation}
0 = (\sDiv_{2s-1,1} \sCurlDagger_{2s,0} \phi)_{A \dots C}. 
\end{equation}
For the case $s=2$, the equation 
$\nabla_{A'}{}^D \Psi_{ABCD} = 0$
is the Bianchi equation, which holds for the Weyl spinor in any vacuum spacetime. Due to the Buchdahl constraint, it holds that in any sufficiently general spacetime, a solution of the spin-2 equation is proportional to the Weyl spinor of the spacetime.  

\subsubsection{Killing spinors}\label{sec:KillingSpinors} 
Spinors $\varkappa_{A_1\cdots A_k}{}^{A_1' \cdots A_k'} \in \SymSpinSec_{k,l}$ satisfying 
$$
(\sTwist_{k,l} \varkappa)_{A_1 \cdots A_{k+1}}{}^{A_1' \cdots A_{k+1}'} = 0,
$$
are called Killing spinors of valence $(k,l)$. 
We denote the space of Killing spinors of valence $(k,l)$ by $\KillSpin_{k,l}$. 
The Killing spinor equation is an over-determined system. 
The space of Killing spinors is a finite dimensional space, and the existence of Killing spinors imposes strong restrictions on $\Mcal$, see section \ref{sec:KillSpinSpacetime} below.  Killing spinors $\GenVec_{AA'} \in \KillSpin_{1,1}$ are simply conformal Killing vector fields, while Killing spinors $\kappa_{AB} \in \KillSpin_{2,0}$ are also known as conformal Killing-Yano forms, or twistor forms.\footnote{In the mathematics literature, Killing spinors of valence $(1,0)$ are known as twistor spinors. The terms conformal Killing-Yano form or twistor form is used also for the real 2-forms corresponding to Killing spinors of valence $(2,0)$, as well as for forms of higher degree and in higher dimension, in the kernel of an analogous Stein-Weiss operator.} 
Further, we mention that  Killing spinors $L_{ABA'B'} \in \KillSpin_{2,2}$ are traceless symmetric conformal Killing tensors $L_{ab}$, satisfying the equation 
$$
0 = \nabla_{(a}L_{bc)} -  \tfrac{1}{3} g_{(ab}\nabla^{d}L_{c)d}
$$
For any $\kappa_{AB} \in \KillSpin_{2,0}$ we have that $L_{ABA'B'} = \kappa_{AB} \bar \kappa_{A'B'} \in \KillSpin_{2,2}$.  
See section \ref{sec:KillSpinSpacetime} below for further details.

\subsection{Space spinors} 
Let $\tau^a$ be a timelike vector, normalized so that $\tau^a \tau_a = 1$.  
Define the projector
\[
\Spaceh_{ab}= g_{ab} -\tau_a\tau_b.  
\]
The space-spinor version of the soldering form is 
$$
\Spaceh_a{}^{AB} = \sqrt{2}\met_a{}^{(A}{}_{A'} \tau^{B)A'}
$$
This gives a correspondence which represents spatial vectors $x^a$ with respect to $\tau^a$, i.e. satisfying $x^a \tau_a = 0$, in terms of symmetric spinors. The Hermitian conjugate of a spinor $\lambda_A$ is defined as 
$$
\hat \lambda_A = \sqrt{2}\tau_A{}^{B'} \bar \lambda_{B'}.
$$
A spinor with even valence is called real if $\hat \lambda_{A_1\dots A_{2s}} = (-1)^s \lambda_{A_1\dots A_{2s}}$. Real spinors of even valence, eg. $\omega_{AB}, \xi_{ABCD}$ correspond to real tensors $\omega_a, \xi_{ab}$.

A general spinor can be decomposed into space spinor terms and terms containing $\tau^{AA'}$. For example, 
\begin{align}
\GenVec_{AA'}={}&\tau_{AA'} \GenVec -  \sqrt{2} \tau^{B}{}_ {A'} \GenVec_{AB},
\end{align}
where $\GenVec=\tau^{AA'} \GenVec_{AA'}$, $\GenVec_{AB}=\sqrt{2} \tau_{(A}{}^{A'}\GenVec_{B)A'}$ are a scalar and a space spinor, respectively.

We also define the second fundamental form as
\begin{align}
\SpaceK_{ab} = \Spaceh_{a}{}^{c} \Spaceh_{b}{}^{d} \nabla_{c}\tau_{d}.
\end{align}

Applying the space spinor split to the spinor covariant derivative $\nabla_{AA'}$ gives 
$$
\nabla_{AA'} = \tau_{AA'} \Dnormal - \sqrt{2} \tau^B{}_{A'}\DSen_{AB} 
$$
where now $\Dnormal = \tau^{AA'} \nabla_{AA'}$ is the normal derivative 
and $\nabla_{AB} = \sqrt{2}\tau_{(A}{}^{A'} \DSen_{B)A'}$  is the Sen connection.

Let $\SpaceK_{ABCD}$ denote the space spinor counterpart of the tensor $\SpaceK_{ab}$.  One
has that 
\[
\SpaceK_{ABCD} = \sqrt{2} \tau_{C}{}^{A'} \nabla_{AB}\tau_{DA'}, \quad \SpaceK_{ABCD}=\SpaceK_{(AB)(CD)}.
\]

For the rest of this section we will assume that $\tau^a$ is the timelike normal of a Cauchy surface $\SpaceSlice$. With a slight abuse of notation we will identify such tensors and spinors on the spacetime with their pullbacks to the surface $\SpaceSlice$. Let $D_a$ denote the intrisic Levi-Civita connection on $\SpaceSlice$, and $D_{AB}=D_{(AB)}=\sigma^a{}_{AB}D_a$ its spinorial counterpart. Then we see that the  Sen connection, $\DSen_{AB}$, and the Levi-Civita connection, $D_{AB}$, are related to each other through the spinor $\SpaceK_{ABCD}$. For example, for a valence 1 spinor $\pi_C$ one has that
\[
\DSen_{AB} \pi_C = D_{AB}\pi_C + \frac{1}{\sqrt{2}} \SpaceK_{ABC}{}^{D} \pi_D.
\]

On the surface $\SpaceSlice$ the Weyl spinor can be split into its electric and magnetic parts via
\[
E_{ABCD} \equiv \frac{1}{2}\left( \Psi_{ABCD} +
  \hat{\Psi}_{ABCD}\right), \quad B_{ABCD}\equiv
\frac{\mbox{i}}{2}\left(\hat{\Psi}_{ABCD} - \Psi_{ABCD} \right),
\]
so that
\[
\Psi_{ABCD} = E_{ABCD} + \mbox{i}B_{ABCD}. 
\]

Crucial for our applications is that the spinors $E_{ABCD}$ and
$B_{ABCD}$ can be expressed in terms of quantities intrinsic to the
hypersurface $\SpaceSlice$. In detail, we have 
\begin{subequations}
\begin{align}
E_{ABCD} ={}& - \tilde{\Phi}_{ABCD} -  \SpaceK^{FH}{}_{FH} \SpaceK_{(ABCD)} -  r_{(ABCD)} + \SpaceK_{(AB}{}^{FH}\SpaceK_{CD)FH}, \label{Weyl:Electric}\\
B_{ABCD} ={}& -i \sqrt{2} D_{(A}{}^{F}\SpaceK_{BCD)F}, \label{Weyl:Magnetic}
\end{align}
\end{subequations}
where $r_{ABCD}$ is the space spinor counterpart of the  Ricci tensor of the intrinsic metric of the hypersurface $\SpaceSlice$, and $\tilde{\Phi}_{ABCD} = 2 \Phi_{(AC|B'D'|}\tau_{B}{}^{B'}\tau_{D)}{}^{D'}$ is given by the matter content.

We can formulate the Cauchy problem for the spin-$s$ testfield equation in terms of space spinors as follows.\footnote{Observe that we do not assume vacuum when we study propagation of spin-$s$ test fields. We only assume that the evolution of the metric is known.}
The space spinor split of the spin-$s$ equation $(\sCurlDagger_{2s,0}\varphi)_{A\dots E A'}=0$
takes the form 
$$
0 = \tfrac{1}{\sqrt{2}} \Dnormal \varphi_{A\dots F} - \DSen_{A}{}^{G}\varphi_{B\dots FG}=0.
$$
If we split this equation into irreducible parts we get the first order, symmetric hyperbolic, evolution equation
$$
\Dnormal\varphi_{A\dots F} = \sqrt{2} \DSen_{(A}{}^{G}\varphi_{B\dots F)G},
$$
and for the cases $s\geq 1 $, the constraint equation 
\begin{equation}
\DSen^{AB} \varphi_{AB\dots F} = 0. \label{eq:constraintspins}
\end{equation}
on $\SpaceSlice$. One can verify that this constraint automatically propagates for $s=1$. For higher spin the Buchdahl constraint gives an obstruction for propagation of the constraint \eqref{eq:constraintspins}.

If we make a space spinor splitting of the valence $(2,0)$ Killing spinor equation $(\sTwist_{2,0} \varkappa)_{ABCA'}=0$, we get
\begin{subequations}
\begin{align}
\Dnormal\varkappa_{AB}={}&- \tfrac{1}{\sqrt{2}}\DSen_{(A}{}^{C}\varkappa_{B)C},\\
\DSen_{(AB}\varkappa_{CD)}={}&0.
\end{align}
\end{subequations}
Hence, also in this case we have an evolution equation and a constraint equation. However, the integrability condition for the Killing spinor gives an obstruction to the propagation of the constraint. The propagation of the integrability condition is a bit more complicated, but we still have the following result for vacuum spacetimes. 
\begin{theorem}[\protect{\cite[Theorem 9]{backdahl:valiente-kroon:2010:MR2753388}, \cite[Theorem 4]{backdahl:valiente-kroon:2012JMP....53d2503B}}]

\label{Theorem:KSData}
Consider an initial data set for the Einstein vacuum field equations on a Cauchy hypersurface $\SpaceSlice$. Let $\mathcal{U}\subset\SpaceSlice$ be an open set. 
The development of the initial data set will then have a Killing spinor in the domain of dependence of $\mathcal{U}$ if and only if 
\begin{subequations}
\begin{align}
\DSen_{(AB}\varkappa_{CD)}={}&0,\\
\Psi_{(ABC}{}^F\varkappa_{D)F}={}&0,
\end{align}
\end{subequations}
are satisfied on $\mathcal{U}$. 
\end{theorem}
Observe that these two conditions can be formulated entirely in terms of the data for $\varkappa_{AB}$, $\SpaceK_{ABCD}$ and $r_{ABCD}$, i.e. quantities intrinsic to the surface $\SpaceSlice$.

\section{Spacetimes with special geometry} \label{sec:specialgeom} 

\subsection{Algebraically special spacetimes} 
Let $\varphi_{A\cdots D} \in \SymSpinSec_{k,0}$. A spinor $\alpha_A$ is a \emph{principal spinor} of $\varphi_{A\cdots D}$ if  
$$
\varphi_{A \cdots D} \alpha^A \cdots \alpha^D= 0.
$$
An application of the fundamental theorem of algebra shows that any $\varphi_{A \cdots D} \in \SymSpinSec_{k,0}$ has exactly $k$ principal spinors $\alpha_A, \dots, \delta_A$, and hence is of the form 
$$
\varphi_{A \cdots D} = \alpha_{(A} \cdots \delta_{D)}.
$$
If $\varphi_{A \cdots D} \in \SymSpinSec_{k,0}$ has $n$ distinct principal spinors $\alpha^{(i)}_A$, repeated $m_i$ times, then $\varphi_{A \cdots D} $ is said to have algebraic type $\{m_1, \dots, m_n\}$. 
Applying this to the Weyl tensor leads to the Petrov classification, see table 1. We have the following list of algebraic, or Petrov, types\footnote{The Petrov classification is exclusive, so a spacetime belongs at each point to exactly one Petrov class.}. 
\begin{table}[!htb]
\begin{center}
\begin{tabular}{l|l|l} 
I & $\{1,1,1,1\}$ & $\Psi_{ABCD} = \alpha_{(A} \beta_B \gamma_C \delta_{D)}$ \\ 
II & $\{2,1,1\}$ & $\Psi_{ABCD} = \alpha_{(A} \alpha_B \gamma_C \delta_{D)}$ \\ 
D & $\{2,2\}$ & $\Psi_{ABCD} = \alpha_{(A} \alpha_B \beta_C \beta_{D)}$ \\ 
III & $\{3,1\}$ & $\Psi_{ABCD} = \alpha_{(A} \alpha_B \alpha_C \beta_{D)}$ \\ 
N & $\{4\}$ & $\Psi_{ABCD} = \alpha_A \alpha_B \alpha_C \alpha_D$ \\ 
O & $\{ - \}$ & $\Psi_{ABCD} = 0$
\end{tabular} 
\caption{The Petrov classification}
\end{center} 
\end{table} 
A principal spinor $o_A$ determines a principal null direction $l_a = o_A \bar o_{A'}$. 
The Goldberg-Sachs theorem states that in a vacuum spacetime, the congruence generated by a null field $l_a$ is geodetic and shear free (i.e. $\sigma = \kappa =0$) if and only if $l_a$ is a repeated principal null direction of the Weyl tensor $C_{abcd}$ (or equivalently $o_A$ is a repeated principal spinor of the Weyl spinor $\Psi_{ABCD}$). 

\subsubsection{Petrov type $\PetrovD$}
The vacuum type $\PetrovD$ spacetimes have been classified by Kinnersley \cite{kinnersley:1969JMP....10.1195K},  see also Edgar et al \cite{edgar:etal:2009CQGra..26j5022E}. 
A Petrov type $\PetrovD$ spacetime has two repeated principal spinors $o_A, \iota_A$. In this case, the Weyl spinor takes the form  
$$
\Psi_{ABCD} = \frac{1}{6} \Psi_2 o_{(A} o_B \iota_C \iota_{D)}.
$$
In this case $\kappa_{AB} \in \KillSpin_{2,0}$ is of the form $\kappa_{AB} = -2 \kappa_1 o_{(A} \iota_{B)}$. In particular, the principal spinors of $\kappa_{AB}$ coincide with the principal spinors of $\Psi_{ABCD}$. One finds, using this fact and the Bianchi identity, that in a vacuum Petrov type $\PetrovD$ spacetime, $\kappa_{AB} \in \KillSpin_{2,0}$ if and only if $\kappa_1 \propto \Psi_2^{-1/3}$; hence the space of Killing spinors is $1$-dimensional. 
Since the Petrov classes are exclusive, we have that $\Psi_2 \ne 0$ for a Petrov type D space.  
It follows from the above that in a vacuum Petrov type $\PetrovD$ spacetime, there is a Killing spinor $\kappa_{AB}$, and the principal spinors of $\kappa_{AB}$ coincide with those of the Weyl spinor $\Psi_{ABCD}$.

\subsection{Killing spinor spacetimes} \label{sec:KillSpinSpacetime} 
Differentiating the Killing spinor equation $(\sTwist_{k,l} \phi)_{A \cdots D A' \cdots D'} = 0$,  and commuting derivatives yields
an algebraic relation between the curvature, Killing spinor, and their covariant derivatives which restrict the curvature spinor, see \cite[\S 2.3]{ABB:symop:2014CQGra..31m5015A}. 
Explicitely, 
for a valence $(1,0)$ Killing spinor $\kappa_A$, we have the condition
\begin{subequations}\label{eq:curv-restrict}
\begin{align}\label{eq:curv-restrict-valence-1}
\Psi_{ABCD} \kappa^D ={}& 0 \\  
\intertext{while for a valence $(2,0)$ Killing spinor $\kappa_{AB}$, the condition takes the form} \label{eq:curv-restrict-valence-2}
\Psi_{(ABC}{}^E \kappa_{D)E} ={}& 0
\end{align}
\end{subequations} 
If $\Mcal$ admits a Killing spinor of valence $(1,0)$, then by \eqref{eq:curv-restrict-valence-1} it is of Petrov type $\PetrovN$ or $\PetrovO$. The vacuum spacetimes of type $\PetrovN$ admitting a Killing spinor of valence $(1,0)$ have been classified by Lewandowski \cite{lewandowski:1991CQGra...8L..11L}. 
Similarly, by \eqref{eq:curv-restrict-valence-2}we have that a spacetime admitting a valence $(2,0)$ Killing spinor is of type $\PetrovD, \PetrovN$, or $\PetrovO$.

\begin{definition} \label{def:alignedmatter} 
A spacetime is said to satisfy the aligned matter condition with respec to $\Psi_{ABCD}$ if
\begin{align}
0={}&\Psi_{(ABC}{}^{F}\Phi_{D)FA'B'}.\label{eq:alignedmatterpsi}
\end{align}
If a spacetime has a valence $(2,0)$ Killing spinor $\kappa_{AB}$, then we say that it satisfies the aligned matter condition with respect to $\kappa_{AB}$, if 
\begin{align}
0={}&\Phi_{(A}{}^{C}{}_{|A'B'|}\kappa_{B)C}.\label{eq:alignedmatter}
\end{align}
\end{definition} 
\begin{remark}
In a spacetime of Petrov type D or N with a valence $(2,0)$ Killing spinor, these two conditions agree, so we can simply say the aligned matter condition.
\end{remark}
\begin{remark} 
The aligned matter condition is interesting since a number of properties of vacuum spacetimes generalize to spacetimes with aligned matter. An example of a spacetime with aligned matter is the Kerr-Newman charged, rotating, black hole solution. 
This metric can be obtained from the Kerr metric \eqref{eq:met} by setting $\Delta = |Q|^2 + a^2 - 2 M r + r^2$, where $Q$ is the electromagnetic charge. Replacing $\Delta$ in  \eqref{eq:KerrTetrad} below by the just given expression, yields a null tetrad for the Kerr-Newman metric.
In geometric units, we have 
\begin{subequations}
\begin{align}
\kappa_{AB} ={}& \tfrac{2}{3} (r - i a \cos\theta) o_{(A}\iota_{B)},\label{eq:kappaKerrNewman}\\
\phi_{AB} ={}& \frac{Q o_{(A}\iota_{B)}}{
(r - i a \cos\theta)^2}
= \frac{\sqrt{2} Q \kappa_{AB}}{9 
\bigl(- (\kappa_{CD} \kappa^{CD})\bigr)^{3/2}} \\
\Phi_{ABA'B'}={}&\frac{2 |Q|^2 o_{(A}\iota_{B)} \bar{o}_{(A'}\bar{\iota}_{B')}}{\Sigma^2}=2 
\phi_{AB} \bar{\phi}_{A'B'}.
 \end{align}
 \end{subequations}
 \end{remark} 

If $(\Mcal, \met_{ab})$ has a valence $(2,0)$ Killing spinor $\kappa_{AB}$ for which the aligned matter condition holds, the 1-form  
\begin{equation}\label{eq:xidef}
\xi_{AA'} = (\sCurlDagger_{2,0} \kappa)_{AA'},
\end{equation} 
is a Killing field, $\nabla_{(a} \xi_{b)} = 0$. To see this, apply a $\sTwist_{1,1}$ to both sides of \eqref{eq:xidef} and commute derivatives. This gives 
\begin{align}
(\sTwist_{1,1} \xi)_{ABA'B'}={}&-3 \Phi_{(A}{}^{C}{}_{|A'B'|}\kappa_{B)C}.
\end{align}
and hence $(\sTwist_{1,1} \xi)_{AA'} = 0$ in case the aligned matter condition holds. Hence $\xi_{AA'}$ is a conformal Killing field. To see that $\xi_{AA'}$ is a Killing field, we note that also $\sDiv_{1,1} \xi = 0$ due to the fact that $\sDiv_{1,1} \sCurlDagger_{2,0} = 0$, which follows from the commutation formulas given in 
\cite[Lemma 18]{ABB:symop:2014CQGra..31m5015A}.

Clearly the real and imaginary parts of $\xi_a$ are also Killing fields. 
If $\xi_a$ is proportional to a real Killing field, we can without loss of generality assume that $\xi_a$ is real.
In this case, the 2-form  
$$
Y_{ab} = \tfrac{3}{2}i (\kappa_{AB}\bar\epsilon_{A'B'}  -  \bar{\kappa}_{A'B'}\epsilon_{AB})
$$
is a Killing-Yano tensor, $\nabla_{(a} Y_{b)c} = 0$, and the symmetric 2-tensor $K_{ab} = Y_a{}^c Y_{cb}$ is a Killing tensor $\nabla_{(a} K_{bc)} = 0$. Further, in this case,
$$
\zeta_a = \xi^b K_{ab} 
$$
is a Killing field. 

\begin{remark} In the case of Kerr, with $\kappa_{AB}$ given by \eqref{eq:kappaKerrNewman}, we get 
\begin{align*}
 \xi^a={}&(\partial_t)^a,\\
\zeta^a={}&a^2 (\partial_t)^a
 + a (\partial_\phi)^a.
 \end{align*}
\end{remark}

\begin{remark} 
\begin{enumerate}
\item 
In the class of vacuum spacetimes of Petrov type $\PetrovD$, the existence of a Killing tensor excludes the 
Kinnersley type III metrics \cite{kinnersley:1969JMP....10.1195K}, see \cite{1976IJTP...15..311C}. 
The complement includes the Kerr-NUT family of spacetimes which thus do admit a Killing tensor. 
Vacuum spacetimes with $\kappa_{AB} \in \KillSpin_{2,0}$ such that $\xi_{AA'}$ is proportional to a real Killing field are said to be in the generalized Kerr-NUT class, see \cite{backdahl:valiente-kroon:2010:MR2753388,backdahl:valiente-kroon:2010PhRvL.104w1102B,ferrando:saez:2007JMP....48j2504F}.

\item 
If $\Mcal$ is a vacuum spacetime of Petrov type $\PetrovN$, then a valence $(2,0)$ Killing spinor factorizes as $\kappa_{AB} = \lambda_A \lambda_B$ where $(\sTwist_{1,0} \lambda)_{ABA'} = 0$. 
This can be shown by comparing the equations for valence $(2,0)$ and valence $(1,0)$ Killing spinors on a vacuum type $\PetrovN$ spacetime in Newman-Penrose formalism, and making use of 
\eqref{eq:curv-restrict}.   
\item 
A Killing spinor of valence $(4,0)$, in a vacuum spacetime factorizes into factors of valence $(2,0)$, see 
\cite[Theorem 8]{ABB:symop:2014CQGra..31m5015A}, see also Remark~\ref{rem:remark5:2}.
For Killing spinors of valence $(k,l)$, the situation is more complicated. 
\item An example of a metric with a Killing spinor of valence $(2,2)$ which does not factorize is given in \cite[\S 6.3]{ABB:symop:2014CQGra..31m5015A}. This metric also satisfies the aligned matter condition.
\end{enumerate}
\end{remark} 

\section{The Kerr spacetime} \label{sec:kerrspacetime} 
The Kerr metric is algebraically special, of Petrov type $\PetrovD$, i.e. there are two repeated principal null directions $\NPl^a, \NPn^a$, for the Weyl tensor, see section \ref{sec:specialgeom}.  We can without loss of generality assume that $\NPl^a \NPn_a = 1$, and define a null tetrad by adding  complex null vectors $\NPm^a, \NPmbar^a$ normalized such that $\NPm^a \NPmbar_a = -1$.
By the Goldberg-Sachs theorem both $\NPl^a, \NPn^a$ are geodetic and shear free, and only one of the 5 independent complex Weyl scalars is non-zero, namely 
\begin{align}
\Psi_2={}&- l^{a} m^{b} \bar{m}^{d} n^{c} C_{abcd}= - \frac{M}{(r -i a \cos\theta)^3}.
\end{align}
An explicit choice for $\NPl^a, \NPn^a, \NPm^a$ is given by the Carter tetrad \cite{znajek:1977MNRAS.179..457Z}  
\begin{subequations}\label{eq:KerrTetrad}
\begin{align}
\NPl^{a}={}&\frac{a (\partial_\phi)^{a}}{\sqrt{2} \KDelta^{1/2} \KSigma^{1/2}}
 + \frac{(a^2 + r^2)(\partial_t)^{a} }{\sqrt{2} \KDelta^{1/2} \KSigma^{1/2}}
 + \frac{\KDelta^{1/2}(\partial_r)^{a} }{\sqrt{2} \KSigma^{1/2}},\\
\NPn^{a}={}&\frac{a (\partial_\phi)^{a}}{\sqrt{2} \KDelta^{1/2} \KSigma^{1/2}}
 + \frac{(a^2 + r^2)(\partial_t)^{a} }{\sqrt{2} \KDelta^{1/2} \KSigma^{1/2}}
 -  \frac{\KDelta^{1/2}(\partial_r)^{a} }{\sqrt{2} \KSigma^{1/2}},\\
\NPm^{a}={}&\frac{(\partial_\theta)^{a}}{\sqrt{2} \KSigma^{1/2}}
 + \frac{i \csc\theta(\partial_\phi)^{a} }{\sqrt{2} \KSigma^{1/2}}
 + \frac{i a \sin\theta (\partial_t)^{a} }{\sqrt{2} \KSigma^{1/2}}.
\end{align}
\end{subequations}
In view of the normalization of the tetrad, the metric takes the form $\met_{ab} = 2 (\NPl_{(a} \NPn_{b)} - \NPm_{(a} \NPmbar_{b)})$. We remark that the choice of $\NPl^a$, $\NPn^a$ to be aligned with the principal null directions of the Weyl tensor, together with the normalization of the tetrad fixes the tetrad up to rescalings. Taking the point of view that the tetrad components of tensors are sections of complex line bundles with action of the non-vanishing complex scalars corresponding to the rescalings of the tetrad, leads to the GHP formalism \cite{GHP}. 

The tensor  
\begin{equation}
K_{ab} = 2\KSigma l_{(a} n_{b)} - r^2 g_{ab} \label{eq:KDefinition}
\end{equation}
is a Killing tensor, satisfying 
$\nabla_{(a} K_{bc)} = 0$. 
For a geodesic $\gamma$, the quantity $\geodcarter = K_{ab} \dot \gamma^a \dot \gamma^b$, known as Carter's constant, is conserved.    
For $a \ne 0$, the tensor $K_{ab}$ cannot be expressed as a tensor product of Killing fields \cite{walker:penrose:1970CMaPh..18..265W}, and similarly Carter's constant $\geodcarter$ cannot be expressed in terms of the constants of the motion associated to Killing fields.  
In this sense $K_{ab}$ and $\geodcarter$ manifest a \emph{hidden symmetry} of the Kerr spacetime. As we shall see in section \ref{sec:hidsym} below, these structures are also related to symmetry operators and separability properties, as well as conservation laws, for field equations on Kerr, and more generally in spacetimes admitting Killing spinors satisfying certain auxiliary conditions.

\newcommand{\geodesic}{\gamma}
\newcommand{\geodesictangent}{\dot\gamma}
\newcommand{\CurlyR}{\mathcal{R}}
\newcommand{\geodq}{\mathbf{q}}

\newcommand{\Sphere}{\mathbb{S}}

\subsection{Geodesics in Kerr} \label{sec:geodesics} 
The dispersive properties of fields, i.e. the tendency of the energy density contained within any stationary region to decrease asymptotically to the future is a crucial property for solutions of field equations on spacetimes, and any proof of stability must exploit this phenomenon. In view of the geometric optics approximation, the dispersive property of fields can be seen in an analogous dispersive property of null geodesics, i.e. the fact that null geodesics in the Kerr spacetime which do not orbit the black hole at a fixed radius must leave any stationary region in at least one of the past or future directions. We will here give an explanation for this fact using tools which can readily be adapted to the case of field equations.

Conserved quantities play a crucial role in understanding the
behaviour of geodesics as well as fields. Along any geodesic $\geodesic^a$ with velocity $\dot \geodesic^a$ in the Kerr spacetime, there are the following conserved quantities
\begin{align*}
\geodmass&=\met_{ab}\geodesictangent^a\geodesictangent^b, &
\geodenergy&=(\partial_t)^a \geodesictangent_a, &
\GeodesicLz&=(\partial_\phi)^a \geodesictangent_a, &
\geodcarter&=K_{ab} \geodesictangent^a\geodesictangent^b,
\end{align*}
which are the mass squared, the
energy, the azimuthal angular momentum, and Carter's fourth constant
respectively. The presence of the extra conserved
quantity allows one to 
integrate
the equations of geodesic motion\footnote{In general, the geodesic equation in a 4-dimensional stationary and axi-symmetric spacetime cannot be integrated, and the dynamics of particles may in fact be chaotic, see \cite{2008PhRvD..77b4035G, 2010PhRvD..81l4005L} and references therein.}.

We shall consider only null geodesics, i.e.  $\geodmass=0$. In this case, it is convenient to
introduce 
\begin{align*}
\geodq
&=\geodcarter -2a\geodenergy\GeodesicLz-\GeodesicLz^2 
=\left(\partial_\theta^a\partial_\theta^b+\frac{\cos^2\theta}{\sin^2\theta}\partial_\phi^a\partial_\phi^b+a^2\sin^2\theta\partial_t^a\partial_t^b
\right)\geodesictangent_a\geodesictangent_b .
\end{align*}
Observe that $\geodq$ is both a conserved quantity, since it is a sum
of conserved quantities, and non-negative, since it is a sum of
non-negative terms.
Of most interest to us is the equation for the $r$-coordinate \cite{FrolovNovikov},
\begin{align}
\KSigma^2 \left ( \frac{\di r}{\di \lambda} \right )^2 =&
-\CurlyR(r;M,a;\GeodesicEnergy,\GeodesicLz,\GeodesicQ) , 
\label{eq:IntroNullGeodesicsr}\\
\intertext{where $\lambda$ is the affine parameter of the null geodesic and }
\CurlyR(r;M,a;\GeodesicEnergy,\GeodesicLz,\GeodesicQ)
=& -(r^2+a^2)^2\GeodesicEnergy^2 -4aMr\GeodesicEnergy\GeodesicLz
+(\KDelta-a^2)\GeodesicLz^2 +\KDelta\GeodesicQ . 
\label{eq:RR:consquant} 
\end{align} 
For a null geodesic $\Geodesic^a$, we define the energy associated with  a vector field $\vecX$ and evaluated on a Cauchy hypersurface $\HypersurfaceGeneral$ to be
\begin{align*}
\GenEnergyGeodesic{\vecX}[\Geodesic](\HypersurfaceGeneral)
&= \gMetric_{\alpha\beta}\vecX^\alpha \dot\Geodesic^\beta|_{\HypersurfaceGeneral}. 
\end{align*}
Since $\nabla_{\dot\Geodesic}\dot\Geodesic=0$ for a geodesic, integrating the derivative of the energy gives  
\begin{align}
\GenEnergyGeodesic{\vecX}[\Geodesic](\HypersurfaceGeneral_2)-\GenEnergyGeodesic{\vecX}[\Geodesic](\HypersurfaceGeneral_1)
=& \int_{\lambda_1}^{\lambda_2} (\dot\Geodesic_\alpha\dot\Geodesic_\beta) \nabla^{(\alpha}\vecX^{\beta)} \di\lambda ,
\label{eq:deformForGeodesics}
\end{align}
where $\lambda_i$ is the unique value of $\lambda$ such that $\Geodesic(\lambda)$ is the intersection of $\Geodesic$ with $\HypersurfaceGeneral_i$. 
Formula \eqref{eq:deformForGeodesics} is particularly easy to work with, if one recalls that 
\begin{align*}
\nabla^{(\alpha}\vecX^{\beta)}&= -\frac12 \Lie_{\vecX}\gMetric^{\alpha\beta} .
\end{align*}
The tensor $\nabla^{(\alpha}\vecX^{\beta)}$ is commonly called the ``deformation tensor''. In the following, unless there is room for confusion, we will drop reference to $\Geodesic$ and $\HypersurfaceGeneral$ in referring to $\GenEnergyGeodesic{\vecX}$.

If one makes the (implicitly defined) change of variables $\di\tau/\di\lambda=\KSigma^{-1}$, then equation \eqref{eq:IntroNullGeodesicsr} for the radial component becomes $(\di r/\di\tau)^2=-\CurlyR(r;M,a;\GeodesicEnergy,\GeodesicLz,\GeodesicQ)$. For fixed $(M,a)$ and $(\GeodesicEnergy,\GeodesicLz,\GeodesicQ)$, this takes the form of the equations of motion of particle in $1$-dimension with a potential. The roots and double roots of the potential $\CurlyR$ determine the turning points and stationary points, respectively, for the motion in the $r$ direction. The potential $-\CurlyR=((r^2+a^2)\GeodesicEnergy+a\GeodesicLz)^2 -\KDelta(\GeodesicQ+\GeodesicLz^2+2a\GeodesicEnergy\GeodesicLz)$ is always non-negative at $r=\rp$ and, unless $\GeodesicEnergy=0$, is positive as $r\rightarrow\infty$, and has at most two roots counting multiplicity. 

By simply considering the turning points, one can use $r$ and $\dot\gamma_r$ to construct a quantity that is increasing overall from the asymptotic past to the asymptotic future. In fact, for a null geodesic with given parameters $(\GeodesicEnergy,\GeodesicLz,\GeodesicQ)$, one may use a simple turning point analysis to show that there is a number $\rorbit \in (\rp, \infty)$ so that the quantity $(r-\rorbit)\dot\Geodesic^r$ increases overall. This quantity corresponds to the energy $\GenEnergyGeodesic{\vecMGeodesic}$ for the vector field $\vecMGeodesic=-(r-\rorbit)\dr$.  
Following this idea, we may now look for a function $\fnMrGeodesic$ which will play the role of $-(r - \rorbit)$, so that for $\vecMGeodesic=\fnMrGeodesic\dr$, the energy $\GenEnergyGeodesic{\vecMGeodesic}$ is non-decreasing for all $\tau$ and not merely non-decreasing overall. For $a\not=0$, both $\rorbit$ and $\fnMrGeodesic$ will necessarily depend on both the Kerr parameters $(M,a)$ and the constants of motion $(\GeodesicEnergy,\GeodesicLz,\GeodesicQ)$; the function $\fnMrGeodesic$ will also depend on $r$, but no other variables. 

We define $\vecMGeodesic=\fnMrGeodesic\dr$ and emphasise to the reader that this is a map from the tangent bundle to the tangent bundle, which is not the same as a standard vector field, which is a map from the manifold to the tangent bundle. To derive a monotonicity formula, we wish to choose $\fnMrGeodesic$ so that $\GenEnergyGeodesic{\vecMGeodesic}$ has a non-negative derivative. We define the covariant derivative of $\vecMGeodesic$ by holding the values of $(\GeodesicEnergy,\GeodesicLz,\GeodesicQ)$ fixed and computing the covariant derivative as if $\vecMGeodesic$ were a regular vector field. Similarly, we define $\Lie_{\vecMGeodesic}\gMetric^{\alpha\beta}$ by fixing the values of the constants of geodesic motion. Since the constants of motion have zero derivative along null geodesics, equation \eqref{eq:deformForGeodesics} remains valid. 

The Kerr metric can be written as 
\begin{align*}
\KSigma\gMetric^{\alpha\beta}&= - \KDelta\dr^\alpha\dr^\beta - \frac{1}{\KDelta}\CurlyR^{\alpha\beta} ,
\end{align*}
where
\begin{align*}
\CurlyR^{\alpha\beta}&= -(r^2+a^2)^2\dt^\alpha\dt^\beta -4aMr \dt^{(\alpha}\dphi^{\beta)} +(\KDelta-a^2)\dphi^\alpha\dphi^\beta +\KDelta\TensorQ^{\alpha\beta}\\
\TensorQ^{\alpha\beta}&=\dtheta^\alpha\dtheta^\beta+\cot^2\theta\dphi^\alpha\dphi^\beta+a^2\sin^2\theta\dt^\alpha\dt^\beta .
\end{align*}
The double contraction of the tensor $\CurlyR^{\alpha\beta}$ with the tangent to a null geodesic gives the potential $\CurlyR(r;M,a;\GeodesicEnergy,\GeodesicLz,\GeodesicQ)=\CurlyR^{\alpha\beta}\dot\gamma_\alpha\dot\gamma_\beta$. If one ignores distracting factors of $\KSigma$, $\KDelta$, their derivatives, and constant factors, one finds that the most important terms in $- \Lie_{\vecMGeodesic}\gMetric^{\alpha\beta}\dot\gamma_\alpha\dot\gamma_\beta$ are 
\begin{align*}
-2(\dr\fnMrGeodesic)\dot\gamma_r\dot\gamma_r +\fnMrGeodesic(\dr\CurlyR^{\alpha\beta})\dot\gamma_\alpha\dot\gamma_\beta =-2(\dr\fnMrGeodesic)\dot\gamma_r\dot\gamma_r +\fnMrGeodesic(\dr\CurlyR).
\end{align*} 
The second term in this sum will be non-negative if $\fnMrGeodesic=\dr\CurlyR(r;M,a;\GeodesicEnergy,\GeodesicLz,\GeodesicQ)$. Recall that the vanishing of $\dr\CurlyR(r;M,a;\GeodesicEnergy,\GeodesicLz,\GeodesicQ)$ is one of the two conditions for orbiting null geodesics. With this choice of $\fnMrGeodesic$, the instability of the null geodesic orbits ensures that, for these null geodesics, the coefficient in the first term, $-2(\dr\fnMrGeodesic)$, will be positive. We can now perform the calculations more carefully to show that this non-negativity holds for all null geodesics. 

Since, up to reparameterization, null geodesics are conformally invariant, it is sufficient to work with the conformally rescaled metric $\KSigma\gMetric^{\alpha\beta}$. Furthermore, since $\gamma$ is a null geodesic, for any function $\fnMpGeodesic$, we may add $\fnMpGeodesic\KSigma\gMetric^{\alpha\beta}\dot\gamma_\alpha\dot\gamma_\beta$ wherever it is convenient. Thus, the change in $\GenEnergyGeodesic{\vecMGeodesic}$ is given as the integral of
\begin{align*}
\KSigma\dot\gamma_\alpha\dot\gamma_\beta \nabla^{(\alpha}\vecMGeodesic^{\beta)}
&=\left(-\frac12\Lie_{\vecMGeodesic}(\KSigma\gMetric^{\alpha\beta}) +\fnMpGeodesic\KSigma\gMetric^{\alpha\beta}\right)\dot\gamma_\alpha\dot\gamma_\beta 
\end{align*}

To progress further, choices of $\fnMrGeodesic$ and $\fnMpGeodesic$ must be made. For the choices we make here, the calculations are straight forward but lengthy.  Let $\fnMna$ and $\fnMnb$ be smooth functions of $r$ and the Kerr parameters $(M,a)$. Let $\DiffCurlyRTilde$ denote $\dr(\frac{\fnMna}{\KDelta}\CurlyR(r;M,a;\GeodesicEnergy,\GeodesicLz,\GeodesicQ))$ and choose $\fnMrGeodesic=\fnMna\fnMnb\DiffCurlyRTilde$ and $\fnMpGeodesic=(1/2)(\dr\fnMna)\fnMnb\DiffCurlyRTilde$. In terms of these functions,
\begin{align}
\KSigma\dot\gamma_\alpha\dot\gamma_\beta \nabla^{(\alpha}\vecMGeodesic^{\beta)}
&=\fnMnb (\DiffCurlyRTilde)^2 -\fnMna^{1/2}\KDelta^{3/2} \left(\dr\left(\fnMnb\frac{\fnMna^{1/2}}{\KDelta^{1/2}}\DiffCurlyRTilde\right)\right) \dot\gamma_r^2. 
\label{eq:GeodesicBulk}
\end{align}
If we now take $\fnMna=\fnMca=\KDelta(r^2+a^2)^{-2}$ and 
$\fnMnb=\fnMcb=(r^2+a^2)^4/(3r^2-a^2)$, then 
\begin{subequations}
\begin{align}
\DiffCurlyRTilde
&= 4Ma\frac{3r^2-a^2}{(r^2+a^2)^3}\GeodesicEnergy\GeodesicLz\\
&\quad-2\frac{r^3-3Mr-a^2r+Ma^2}{(r^2+a^2)^3}\GeodesicLz^2
-2\frac{r^3-3Mr+a^2r+Ma^2}{(r^2+a^2)^3}\GeodesicQ ,\nonumber\\
\dr\left(\fnMnb\frac{\fnMna^{1/2}}{\KDelta^{1/2}}\DiffCurlyRTilde\right)
&=-2\frac{3r^4+a^4}{(3r^2-a^2)^2}\GeodesicLz^2
  -2\frac{3r^4-6a^2r^2-a^4}{(3r^2-a^2)^2}\GeodesicQ.\label{eq:GeodesicsDDiffCurlyRTTilde}
\end{align}
\end{subequations}
Since $\GeodesicQ$ is non-negative  it follows that the right-hand side of \eqref{eq:GeodesicsDDiffCurlyRTTilde} is non-positive and that the right-hand side of equation \eqref{eq:GeodesicBulk} is non-negative. 
Since equation \eqref{eq:GeodesicBulk} gives the rate of change, the energy $\GenEnergyGeodesic{\vecMGeodesic}$ is monotone. 

These calculations reveal useful information about the geodesic motion. The positivity of the term on the right-hand side of \eqref{eq:GeodesicsDDiffCurlyRTTilde} shows that $\DiffCurlyRTilde$ can have at most one root, which must be simple. In turn, this shows that $\CurlyR$ can have at most two roots, as previously asserted. 

The role of orbiting geodesics can be seen in equation \eqref{eq:GeodesicBulk}. Along null geodesics for which $\CurlyR$ has a double root, the double root occurs at the root of $\DiffCurlyRTilde$, so it is convenient to think of the corresponding value of $r$ as being $\rorbit$. In particular, this root is where null geodesics orbit the black hole with a constant value of $r$. The first term in \eqref{eq:GeodesicBulk} vanishes at the root of $\DiffCurlyRTilde$, as it must so that $\GenEnergyGeodesic{\vecMGeodesic}$ can be constantly zero on the orbiting null geodesics. When $a=0$, the quantity $\DiffCurlyRTilde$ reduces to $-2(r-3M)r^{-4}(\GeodesicLz^2+\GeodesicQ)$, so that the orbits occur at $r=3M$. The continuity in $a$ of $\DiffCurlyRTilde$ guarantees that its root converges to $3M$ as $a\rightarrow 0$ for fixed $(\GeodesicEnergy,\GeodesicLz,\GeodesicQ)$. 

From the geometrics optics approximation, it is natural to imagine that the monotone quantity constructed in this section for null geodesics might imply the existence of monotone quantities for fields, which would imply some form of dispersion. For the wave equation, this is true. In fact, the above discussion, when carried over to the case of the wave equation, closely parallels the proof of the Morawetz estimate for the wave equation given in \cite{andersson:blue:0908.2265}. The quantity 
$(\dot\Geodesic_\alpha\dot\Geodesic_\beta)(\nabla^{(\alpha}\vecX^{\beta)}) $ corresponds to the Morawetz density, i.e. the divergence of the momentum corresponding to the Morawetz vector field. The role of the conserved quantities $(\GeodesicEnergy,\GeodesicLz,\GeodesicQ)$ for geodesics is played, in the case of fields, by the energy fluxes defined via second order symmetry operators corresponding to these conserved quantities. 
The fact that the quantity $\CurlyR$ vanishes quadratically on the trapped orbits is reflected in the Morawetz estimate for fields, by a quadratic degeneracy of the Morawetz density at the trapped orbits.

\subsection{Characterizations of Kerr} \label{sec:characterization} 

Consider a vacuum Cauchy data set $(\SpaceSlice, \Spaceh_{ij}, \SpaceK_{ij})$. We say that $(\SpaceSlice, \Spaceh_{ij}, \SpaceK_{ij})$ is asymptotically flat if $\SpaceSlice$ has an end $\Re^3 \setminus B(0, R)$ with a coordinate system $(x^i)$ such that 
\begin{equation}\label{eq:asymptflat} 
\Spaceh_{ij} = \delta_{ij} + O_\infty(r^{\alpha}) , \quad \SpaceK_{ij} = O_\infty(r^{\alpha -1})
\end{equation} 
for some $\alpha < - 1/2$.  The Cauchy data set $(\SpaceSlice, \Spaceh_{ij}, \SpaceK_{ij})$ is asymptotically Schwarzschildean if 
\begin{subequations}\label{eq:boosteddecay}
\begin{eqnarray}
&& h_{ij} = -\left(1+\frac{2A}{r}\right)\delta_{ij} - \frac{\alpha}{r}\left( \frac{2x_ix_j}{r^2}-\delta_{ij}\right)+o_\infty(r^{-3/2}), \label{BoostedDecay1} \\
&& k_{ij} = \frac{\beta}{r^2}\left( \frac{2x_ix_j}{r^2}-\delta_{ij} \right) + o_\infty(r^{-5/2}),  \label{BoostedDecay2}
\end{eqnarray}
\end{subequations}
where $A$ is a constant, and $\alpha,\beta$ are functions on $S^2$, see \cite[\S 6.5]{backdahl:valiente-kroon:2010:MR2753388} for details. 
Here, the symbols 
$\ordo_\infty(r^{\alpha})$ are defined in terms of weighted Sobolev spaces, see \cite[\S 6.2]{backdahl:valiente-kroon:2010:MR2753388} for details.

If $(\Mcal, \met_{ab})$ is vacuum and contains a Cauchy surface $(\SpaceSlice, \Spaceh_{ij}, \SpaceK_{ij})$ satisfying \eqref{eq:asymptflat} or \eqref{eq:boosteddecay}, then $(\Mcal, \met_{ab})$ is asymptotically flat, respectively asymptotically Schwarzschildean, at spatial infinity.
In this case there is a spacetime coordinate system $(x^\alpha)$ such that $\met_{\alpha\beta}$ is asymptotic to the Minkowski line element with asymptotic conditions compatible with \eqref{eq:boosteddecay}. 
For such spacetimes, the ADM 4-momentum $P^\mu$ is well defined. The positive mass theorem states that $P^\mu$ is future directed causal $P^\mu P_\mu \geq 0$ (where the contraction is in the asymptotic Minkowski line element), $P^0 \geq 0$, and gives conditions under which $P^\mu$ is strictly timelike. This holds in particular if $\SpaceSlice$ contains an apparent horizon. 

Mars \cite{mars:2000CQGra..17.3353M}  
has given a characterization of the Kerr spacetime as an asymptotically flat vacuum spacetime with a Killing field $\xi^a$ asymptotic to a time translation, positive mass, and an additional condition on the Killing form $F_{AB} = (\sCurl_{1,1} \xi)_{AB}$, 
$$
\Psi_{ABCD} F^{CD} \propto F_{AB}
$$
A characterization in terms of algebraic invariants of the Weyl tensor has been given by Ferrando and Saez \cite{ferrando:saez:2009CQGra..26g5013F}. 
The just mentioned characterizations are in terms of spacetime quantities. The fact that Killing spinor initial data propagates, (see Theorem~\ref{Theorem:KSData})  
can be used to formulate a characterization of Kerr in terms of Cauchy data, see \cite{backdahl:valiente-kroon:2010PhRvL.104w1102B, backdahl:valiente-kroon:2010:MR2753388, backdahl:valente-kroon:2011RSPSA.467.1701B, backdahl:valiente-kroon:2012JMP....53d2503B}

We here give a characterization in terms spacetimes admitting a Killing spinor of valence $(2,0)$. 
\begin{theorem} \label{thm:kerrchar} 
Assume that $(\Mcal, \met_{ab})$ is vacuum, asymptotically Schwarzschildean at spacelike infinity, and contains a Cauchy slice bounded by an apparent horizon. Assume further $(\Mcal, \met_{ab})$ admits a non-vanishing Killing spinor $\kappa_{AB}$ of valence $(2,0)$. Then $(\Mcal, \met_{ab})$ is locally isometric to the Kerr spacetime. 
\end{theorem} 
\begin{proof} 
Let $P^\mu$ be the ADM 4-momentum vector for $\Mcal$. 
By the positive mass theorem, $P^\mu P_\mu \geq 0$. In the case where $\Mcal$ contains a Cauchy surface bounded by an apparent horizon, then $P^\mu P_\mu > 0$ by  
\cite[Remark 11.5]{chrusciel:bartnik:2003math......7278C}\footnote{Section 11 appears only in the ArXiv version of \cite{chrusciel:bartnik:2003math......7278C}.}. 

Recall that a spacetime with a Killing spinor of valence $(2,0)$ is of Petrov type $\PetrovD, \PetrovN$, or $\PetrovO$. 
From asymptotic flatness and the positive mass theorem, we have $C_{abcd} C^{abcd} = O(1/r^6)$, and hence there is a neighbourhood of spatial infinity where $\Mcal$ is Petrov type $\PetrovD$. It follows that near spatial infinity,  
$\kappa_{AB} = -2 \kappa_1 o_{(A} \iota_{B)}$, with $\kappa_1 \propto \Psi_2^{-1/3} = O(r)$. 
It follows from our asymptotic conditions that the Killing field 
$\xi_{AA'} = (\sCurlDagger_{2,0}\kappa)_{AB}$ is $O(1)$ and hence asymptotic to a translation, $\xi^\mu \to A^\mu$ as $r \to \infty$, for some constant vector $A^\mu$. It follows from the discussion in \cite[\S 4]{aksteiner:andersson:2013CQGra..30o5016A} that $A^\mu$ is non-vanishing. 
Now, by \cite[\S III]{beig:chrusciel:1996JMP....37.1939B}, it follows that in the case $P^\mu P_\mu > 0$, then $A^\mu$ is proportional to $P^\mu$, see also \cite{beig:omurchadha:1987AnPhy.174..463B}.
We are now in the situation considered in the work by 
B\"ackdahl and Valiente-Kroon, see \cite[Theorem B.3]{backdahl:valente-kroon:2011RSPSA.467.1701B}, and hence we can conclude that $(\Mcal, \met_{ab})$ is locally isometric to the Kerr spacetime. 
\end{proof}  

\begin{remark}
\begin{enumerate} 
\item 
This result can be turned into a characterization in terms of Cauchy data along the lines in \cite{backdahl:valiente-kroon:2010:MR2753388} using Theorem~\ref{Theorem:KSData}.
\item Theorem \ref{thm:kerrchar} can be viewed as a variation on the Kerr characterization given in 
\cite[Theorem B.3]{backdahl:valente-kroon:2011RSPSA.467.1701B}. In the version given here, the asymptotic conditions on the Killing spinor have been removed.  
\end{enumerate} 
\end{remark}

\section{Hidden symmetries} \label{sec:hidsym}

\subsection{Symmetry operators}  \label{sec:symop} 
A symmetry operator for a field equation is an operator which takes solutions to solutions.  In the paper \cite{andersson:blue:0908.2265}, two of the authors have given a proof of a Morawetz estimate for the scalar wave equation in the Kerr spacetime, which makes use of higher-order conserved currents constructed from the scalar field, using second order symmetry operators related to the Carter constant. In order to generalize this approach to higher spin fields, it is important to gain an understanding of the symmetry operators for this case.  
In the paper \cite{ABB:symop:2014CQGra..31m5015A} we have given a complete characterization of those spacetimes admitting symmetry  operators of second order for the field equations of spins $0, 1/2, 1$, i.e. the conformal wave equation, the Dirac-Weyl equation and the Maxwell equation, respectively, and given the general form of the symmetry operators, up to equivalence. In order to simplify the presentation here, we shall discuss only the spin-$1$ case, and restrict to spacetimes admitting a valence $(2,0)$ Killing spinor $\kappa_{AB}$. 

There are two spin-$1$ equations (left and right) depending on the helicity  of the spinor. These are 
$$ 
(\sCurlDagger_{2,0} \phi)_{AA'} = 0 \quad \text{(left), \quad and } \quad (\sCurl_{0,2} \varphi)_{AA'} = 0 \quad \text{ (right)}
$$
 The real Maxwell equation $\nabla^a F_{ab} = 0$, $\nabla_{[a} F_{bc]} = 0$ for a real two form $F_{ab} = F_{[ab]}$ is equivalent to either the right or the left Maxwell equations. 
Henceforth we will always assume that $\phi_{AB}$ solves the left Maxwell equation.

Given a conformal Killing vector $\GenVec^{AA'}$, we follow \cite[Equations (2) and (15)]{anco:pohjanpelto:2003:CRM:MR2056970}, see also \cite{anco:pohjanpelto:2003:ProcRSoc:MR1997098},
and define a conformally weighted Lie derivative acting on a symmetric valence $(2s,0)$ spinor field as follows 
\begin{definition}
For $\GenVec^{AA'} \in \ker \sTwist_{1,1}$, and $\varphi_{A_1\dots A_{2s}}\in \mathcal{S}_{2s,0}$, we define
\begin{align}
\hat{\mathcal{L}}_{\GenVec}\varphi_{A_1\dots A_{2s}}\equiv{}&\GenVec^{BB'} \nabla_{BB'}\varphi_{A_1\dots A_{2s}}+s \varphi_{B(A_2\dots A_{2s}} \nabla_{A_1)B'}\GenVec^{BB'}
 + \tfrac{1-s}{4} \varphi_{A_1\dots A_{2s}} \nabla^{CC'}\GenVec_{CC'}.
\end{align}
\end{definition}
If $\GenVec^a$ is a conformal Killing field, then 
$(\sCurlDagger_{2,0} \hat{\mathcal{L}}_\GenVec \varphi)_{AA'} = \mathcal{L}_\GenVec (\sCurlDagger_{2,0} \varphi)_{AA'}$.
It follows that the first order operator $\varphi \to \hat{\mathcal{L}}_\GenVec \varphi$ 
defines a symmetry operator of first order, which is also of the first kind.  For the equations of spins $0$ and $1$, the only first order symmetry operators are given by conformal Killing fields. For the spin-$1$ equation, we may have symmetry operators of the first kind, taking left fields to left, i.e. $\ker \sCurlDagger \mapsto \ker \sCurlDagger$ and of the second kind, taking left fields to right, $\ker \sCurlDagger \mapsto \ker \sCurl$.  
Observe that symmetry operators of the first kind are linear symmetry operators in the usual sense, while symmetry operators of the second kind followed by complex conjugation gives anti-linear symmetry operators in the usual sense.

Recall that the Kerr spacetime admits a constant of motion for geodesics $\geodcarter$ which is not reducible to the conserved quantities defined in terms of Killing fields, but rather is defined in terms of a Killing tensor. Similarly, in a spacetime with Killing spinors, the geometric field equations may admit symmetry operators of order greater than one, not expressible in terms of the symmetry operators defined in terms of (conformal) Killing fields. We refer to such symmetry operators as ``hidden symmetries''.  

In general, the existence of symmetry operators of the second order implies the existence of Killing spinors (of valence $(2,2)$ for the conformal wave equation and for Maxwell symmetry operators of the first kind for Maxwell, or $(4,0)$ for Maxwell symmetry operators for of the second kind) satisfying certain auxiliary conditions. The conditions given in \cite{ABB:symop:2014CQGra..31m5015A} are are valid in arbitrary 4-dimensional spacetimes, with no additional conditions on the curvature. As shown in \cite{ABB:symop:2014CQGra..31m5015A}, the existence of a valence $(2,0)$ Killing spinor is a sufficient condition for the existence of second order symmetry operators for the spin-$s$ equations, for $s=0,1/2,1$. 

\begin{remark} \label{rem:remark5:2}
\begin{enumerate} 
\item If $\kappa_{AB}$ is a Killing spinor of valence $(2,0)$, then $L_{ABA'B'} =  \kappa_{AB} \bar \kappa_{A'B'}$ and $L_{ABCD} = \kappa_{AB} \kappa_{CD}$ are Killing spinors of valence $(2,2)$ and $(4,0)$ satisfying the auxiliary conditions given in \cite{ABB:symop:2014CQGra..31m5015A}. 

\item In the case of aligned matter with respect to $\Psi_{ABCD}$, any valence $(4,0)$ Killing spinor $L_{ABCD}$ factorizes, i.e. $L_{ABCD} = \kappa_{AB} \kappa_{CD}$ for some Killing spinor $\kappa_{AB}$ of valence $(2,0)$ \cite[Theorem 8]{ABB:symop:2014CQGra..31m5015A}. An example of a spacetime with aligned matter which admits a valence $(2,2)$ Killing spinor that does not factorize is given in \cite[\S 6.3]{ABB:symop:2014CQGra..31m5015A}, see also \cite{michel:radoux:silhan:2013arXiv1308.1046M}. 
\end{enumerate} 
\end{remark}

\begin{proposition}[\protect{\cite{ABB:currents}}]
\label{prop:SymOpPot}
\begin{enumerate} 
\item The general symmetry operator of the first kind for the Maxwell field, of order at most two, is of the form 
\begin{align}
\chi_{AB}={}&Q \phi_{AB}+(\sCurl_{1,1} A)_{AB}, \label{eq:SymFirstPot} 
\end{align} 
where $\phi_{AB}$ is a Maxwell field, and $A_{AA'}$ is a  linear concomitant\footnote{A concomitant is a covariant, local partial differential operator.} of first order, such that $A_{AA'} \in \ker \sCurlDagger_{1,1}$ and $Q \in \ker \sTwist_{0,0}$, i.e. locally constant.
\item
The general symmetry operator of the second kind for the Maxwell field is of the form 
\begin{align}
\omega_{A'B'}={}&(\sCurlDagger_{1,1} B)_{A'B'}, \label{eq:SymSecondPot} 
\end{align} 
where $B_{AA'}$ is a first order linear concomitant of $\phi_{AB}$ such that $B_{AA'} \in \ker\sCurl_{1,1}$.
\end{enumerate} 
\end{proposition}

\begin{remark} 
The operators $\sCurlDagger_{1,1}$ and $\sCurl_{1,1}$ are the adjoints of the left and right Maxwell operators $\sCurlDagger_{2,0}$ and $\sCurl_{0,2}$.  
As we shall see in section \ref{sec:conscurr} below, conserved currents for the Maxwell field can be characterized in terms of solutions of the adjoint Maxwell equations
\begin{subequations}\label{eq:AdjointMaxwell}
\begin{align} 
(\sCurlDagger_{1,1} A)_{A'B'} &= 0 \label{eq::LeftAdjointMaxwell} \\
(\sCurl_{1,1}B)_{AB} &= 0 \label{eq::RightAdjointMaxwell} 
 \end{align} 
 \end{subequations} 
\end{remark}

\begin{definition}
Given a spinor 
$\kappa_{AB} \in \SymSpinSec_{2,0}$ we define the operators $\sExt_{2,0}: \SymSpinSec_{2,0}\rightarrow \SymSpinSec_{2,0}$ and $\bar{\sExt}_{0,2}: \SymSpinSec_{0,2}\rightarrow \SymSpinSec_{0,2}$ by
\begin{subequations}
\begin{align}
(\sExt_{2,0}\varphi)_{AB}={}&-2 \kappa_{(A}{}^{C}\varphi_{B)C},\\
(\bar{\sExt}_{0,2}\phi)_{A'B'}={}&-2 \bar{\kappa}_{(A'}{}^{C'}\phi_{B')C'}.
\end{align}
\end{subequations}
\end{definition}
Let $\kappa_i$ be the Newman-Penrose scalars for $\kappa_{AB}$. If $\kappa_{AB}$ is of algebraic type $\{1,1\}$ then $\kappa_0 = \kappa_2 = 0$, in which case $\kappa_{AB} = - 2 \kappa_1 o_{(A} \iota_{B)}$. A direct calculations gives the following result. 
\begin{lemma} \label{lem:Ephi-extreme} 
Let $\kappa_{AB} \in \SymSpinSec_{2,0}$ and assume that $\kappa_{AB}$ is of algebraic type $\{1,1\}$. Then the operators $\sExt_{2,0}, \bar{\sExt}_{2,0}$ remove the middle component and rescale the extreme components as 
\begin{subequations}
\begin{align}
(\sExt_{2,0}\varphi)_{0}={}&-2 \kappa_1 \varphi_0,&
(\sExt_{2,0}\varphi)_{1}={}&0,&
(\sExt_{2,0}\varphi)_{2}={}&2 \kappa_1 \varphi_2,\\
(\bar{\sExt}_{0,2}\phi)_{0'}={}&-2 \bar{\kappa}_{1'} \phi_{0'},&
(\bar{\sExt}_{0,2}\phi)_{1'}={}&0,&
(\bar{\sExt}_{0,2}\phi)_{2'}={}&2 \bar{\kappa}_{1'} \phi_{2'}.
\end{align}
\end{subequations}
\end{lemma}
\begin{remark} 
If $\kappa_{AB}$ is a Killing spinor in a Petrov type $\PetrovD$  spacetime, then $\kappa_{AB}$ is of algebraic type $\{1,1\}$. 
\end{remark}

\begin{definition} 
Define the first order 1-form linear concomitants $A_{AA'}, B_{AA'}$ by 
\begin{subequations}\label{eq:ABdef} 
\begin{align} 
A_{AA'}[\kappa_{AB}, \phi_{AB} ] ={}&- \tfrac{1}{3} (\sExt_{2,0} \phi)_{AB}  (\sCurl_{0,2} \bar{\kappa})^{B}{}_{A'}
 + \bar{\kappa}_{A'B'} (\sCurlDagger_{2,0} \sExt_{2,0} \phi)_{A}{}^{B'},\label{eq:Adef} \\
A_{AA'}[\GenVec_{AA'}, \phi_{AB} ] ={}&  \GenVec_{BA'} \phi_{A}{}^{B} \\
B_{AA'}[\kappa_{AB}, \phi_{AB} ] ={}&\kappa_{AB} (\sCurlDagger_{2,0} \sExt_{2,0} \phi)^{B}{}_{A'}
 + \tfrac{1}{3} (\sExt_{2,0} \phi)_{AB} (\sCurlDagger_{2,0} \kappa)^{B}{}_{A'}, \label{eq:Bdef} 
 \end{align}
\end{subequations} 
\end{definition} 
When there is no room for confusion, we suppress the arguments, and write simply $A_{AA'}, B_{AA'}$. 
The following result shows that $A_{AA'}, B_{AA'}$ solves the adjoint Maxwell equations, 
provided $\phi_{AB}$ solves the Maxwell equation. 
\begin{lemma}[\protect{\cite[\S 7]{ABB:symop:2014CQGra..31m5015A}}] 
\label{lem:ABpot}
Assume that $\kappa_{AB}$ is a Killing spinor of valence $(2,0)$, that $\GenVec_{AA'}$ is a conformal Killing field, and that $\phi_{AB}$ is a Maxwell field. Then, with $A_{AA'}, B_{AA'}$ given by \eqref{eq:ABdef} it holds that $A_{AA'}[\kappa_{AB}, \phi_{AB}]$ and $A_{AA'}[\GenVec_{AA'}, \phi_{AB}]$ satisfy $(\sCurlDagger_{1,1} A)_{A'B'} = 0$, and $B_{AA'}[\kappa_{AB}, \phi_{AB}]$ satisfies $(\sCurl_{1,1} B)_{AB} = 0$. 
\end{lemma} 
\begin{remark} Proposition \ref{prop:SymOpPot} together with Lemma \ref{lem:ABpot} show that the existence of a valence $(2,0)$ Killing spinor implies that there are non-trivial second order symmetry operators of the first and second kind for the Maxwell equation. 
\end{remark}

\subsection{Conserved currents} \label{sec:conscurr} 
Recall that the symmetric stress energy tensor for the Maxwell field is $T_{ab} = \phi_{AB} \bar \phi_{A'B'}$. Since the Maxwell equation is conformally invariant, we have $T^a{}_a = 0$.  
If $\phi_{AB}$ solves the Maxwell equation, then $T_{ab}$ is conserved, $\nabla^a T_{ab} = 0$, and hence if $\GenVec^a$ is a conformal Killing field, then the current
$J_a = T_{ab} \GenVec^b$ is conserved, $\nabla^a J_a = 0$. 
Lie differentiating with respect to conformal Killing fields and using a polarized form of the stress energy tensor yields conserved currents which are higher-order in derivatives of the field. 
However, as discovered by Lipkin \cite{lipkin:1964:MR0162484} and Fushchich and Nikitin, see  \cite{fushchich:nikitin:1992:JPhysA:MR1154854} and references therein, there are nontrivial currents for the Maxwell field on Minkowski space which are not given by this construction. 
The conserved currents $J_a$ for the Maxwell and more generally spin-$s$ fields, $s \in \half \NatNum$, on Minkowski space, have been classified by Anco and Pohjanpelto, see \cite{anco:pohjanpelto:2003:ProcRSoc:MR1997098} and references therein\footnote{\label{foot:ancokerr} The problem of classifying conserved currents for the Maxwell field on the Kerr spacetime has been mentioned but not addressed by Anco et al, cf. \cite[p. 55]{anco:the:2005:AAM:MR2220197} and \cite[\S VII]{anco:pohjanpelto:2001:MR1885280}.}.  The conserved currents $J_a$ considered in the just cited works are bilinear or quadratic concomitants of the Maxwell field, of any finite order. The order of such a current is defined to be the sum of the order of derivatives on each factor. Thus, for example, the order of the current $\phi_{AB} \bar \chi_{A'B'} \xi^{BB'}$, where $\chi_{AB}$ is given by \eqref{eq:SymFirstPot}, is two. 

\begin{definition} \label{def:trivcurr} 
A current $\tilde J^{AA'}$ is called \emph{trivial} if it is of the form 
$$
\tilde J^{AA'}= (\sCurlDagger_{2,0} S)_{AA'} + (\sCurl_{0,2} T)_{AA'}
$$
for some symmetric spinor fields $S_{AB}$ and $T_{A'B'}$. 
\end{definition} 
In this case, $(*J)_{abc}$ is an exact 3-form, so the flux through a hypersurface of a trivial current $\tilde J^{AA'}$ is given by a pure boundary term.
We shall consider \emph{equivalence classes} of currents up to trivial currents. Two currents $J^{AA'}, K^{AA'}$ are said to be equivalent if $J^{AA'} - K^{AA'}$ is a trivial current. In this case we write $J^{AA'} \sim K^{AA'}$. 

A current which is invariant under $\phi_{AB} \mapsto i \phi_{AB}$ is said to be of even parity, while a current which changes sign under this substitution is said to be of odd parity, often termed chiral. A current which is a concomitant of
$\phi_{AB}$ can be written as a sum of terms with even or odd parity. 
            
The structure of conserved currents of up to second order 
for the Maxwell field on a general spacetime has recently 
been determined by the authors, see \cite{ABB:currents}. 
As shown in \cite{anco:pohjanpelto:2001:MR1885280},
see also \cite{anco:pohjanpelto:2003:ProcRSoc:MR1997098},
the conserved currents for the Maxwell field on Minkowski space 
are all generated from solutions of
the adjoint equation.  The same statement holds for currents up to second order in a general spacetime \cite{ABB:currents}, and it seems reasonable to conjecture that this holds for currents of arbitrary order.
\begin{lemma}[\protect{\cite{ABB:currents}}]
\label{lemma:BilinearCurrent}
Let $\phi_{AB}$ be a Maxwell field, i.e. $\phi_{AB} \in \ker \sCurlDagger_{2,0}$ and 
assume that $J^{AA'}\in \ker \sDiv_{1,1}$ is a conserved concomitant of $\phi_{AB}$ of quadratic type. 
\begin{enumerate} 
\item If $J_{AA'}$ has even parity, then
\begin{equation*}
J^{AA'} \sim A^{A}{}_{B'}\bar\phi^{A'B'},
\end{equation*}
where $A_{AA'}$ is a linear concomitant of $\phi_{AB}$ satisfying the left adjoint Maxwell equation $(\sCurlDagger_{1,1}A)_{A'B'}=0$.
\item If $J_{AA'}$ has odd parity (chiral), then
\begin{equation*}
J^{AA'} \sim B_{B}{}^{A'}\phi^{AB}+ \overline{\widetilde B}_{B'}{}^{A}\bar\phi^{A'B'}  , 
\end{equation*}
where $B_{AA'}$, ${\widetilde B}_{AA'}$ are linear concomitants of $\phi_{AB}$ satisfying the right adjoint Maxwell equation $(\sCurl_{1,1}B)_{A'B'}=0$, $(\sCurl_{1,1}\widetilde B)_{A'B'}=0$.
\end{enumerate} 

\end{lemma}

\begin{definition}
The stress, zilch and
chiral 
currents are defined in terms of the spinors $A_{AA'}$, $B_{AA'}$ by 
\begin{subequations}\label{eq:quadraticcurrent}
\begin{align} 
\text{ stress: } && \Psi_{S AA'} &= \half (\bar{A}_{A'B} \phi_A{}^B + A_{AB'} 
\bar\phi_{A'}{}^{B'})    \\ 
\text{ zilch: } && \Psi_{Z AA'} &= \half i (\bar{A}_{A'B} \phi_A{}^B - A_{AB'} 
\bar\phi_{A'}{}^{B'} )  \\ 
\text{ chiral: } && \Psi_{C AA'} &= \half(B_{BA'} \phi_A{}^B + \bar B_{B'A} \bar\phi_{A'}{}^{B'}
) 
\end{align}
\end{subequations} 
\end{definition}
Of these, the currents $\Psi_{S AA'}, \Psi_{Z AA'}$ have even parity, while $\Psi_{C AA'}$ has odd parity.

\begin{example} \label{ex:Adj}
\begin{enumerate} 
\item 
Let $A_{AA'}[\GenVec_{AA'}, \phi_{AB}] =  \GenVec_{BA'} \phi_{A}{}^{B}$ 
where $\GenVec_{AA'}$ is a real Killing vector. The current $\Psi_{S AA'} = T_{AA'BB'}\GenVec^{BB'}$ is the standard stress-energy current associated with $\GenVec^a$. 
\item \label{point:Adjchi}
If we have a symmetry operator $\phi_{AB}\rightarrow \chi_{AB}$, the concomitant $A_{AA'}[\GenVec_{AA'}, \chi_{AB}]$ is again a solution of the adjoint equation \eqref{eq::LeftAdjointMaxwell}, and hence the current 
\begin{equation*}
\Psi_{S AA'} = 
\tfrac{1}{2}\GenVec^{BB'} \chi_{AB}\bar{\phi}_{A'B'}+\tfrac{1}{2}\GenVec^{BB'} \phi_{AB}\bar\chi_{A'B'},
\end{equation*} 
is also conserved. The current $\Psi_{S AA'}$ is in this case derived from the polarized form of the standard Maxwell stress energy tensor.  
\end{enumerate} 
\end{example}

\begin{lemma}[\protect{\cite{ABB:currents,2014arXiv1412.2960A}}] Let $\kappa_{AB} \in \KillSpin_{2,0}$, and assume that the aligned matter condition holds with respect to $\kappa_{AB}$. Define  
\begin{align}
\xi_{AA'}\equiv{}&(\sCurlDagger_{2,0} \kappa)_{AA'}, \label{eq:xikappadef} \\
\eta_{AA'}\equiv{}&(\sCurlDagger_{2,0} \sExt_{2,0}\phi)_{AA'}. \label{eq:etadef}
\end{align}
Let 
\begin{align} 
J^1_{AA'} ={}& 
\tfrac{1}{2} \xi^{BB'} \bar{\phi}_{A'B'} \chi_{AB} +  \tfrac{1}{2} \xi^{BB'} \phi_{AB} \bar{\chi}_{A'B'},\label{eq:J1def}\\
J^2_{AA'} ={}&V_{ABA'B'} \xi^{BB'},\label{eq:J2def}
\end{align} 
where $\chi_{AB}$ is given by \eqref{eq:SymFirstPot} with $A_{AA'}$ given by \eqref{eq:Adef} and $Q=0$, and 
\begin{align}
V_{ABA'B'}\equiv{}&\tfrac{1}{2} \eta_{AB'} \bar{\eta}_{A'B}
 + \tfrac{1}{2} \eta_{BA'} \bar{\eta}_{B'A}
 + \tfrac{1}{3} (\sExt_{2,0}\phi)_{AB} (\hat{\mathcal{L}}_{\bar\xi}\bar{\phi})_{A'B'}
 + \tfrac{1}{3} (\bar{\sExt}_{2,0} \bar{\phi})_{A'B'} (\hat{\mathcal{L}}_{\xi}\phi)_{AB}. \label{eq:Vdef}
\end{align}
Then both currents $J^1_{AA'}$ and $J^2_{AA'}$ are conserved.
If we furthermore assume that $\xi_{AA'}$ is real, one can show that the currents are equivalent, up to sign. In detail we get
\begin{align*}
- J^1_{AA'} 
={}&J^2_{AA'}
 + (\sCurl_{0,2} \bar{S})_{AA'}
 + (\sCurlDagger_{2,0} S)_{AA'},
\end{align*}
where
\begin{align}
S_{AB}={}& \tfrac{1}{2} \bar{\eta}^{A'C} \xi^{D}{}_{A'} \kappa_{(AB}\phi_{CD)}
 -  \tfrac{1}{6} \xi^{CA'} (\bar{\sExt}_{0,2}\bar{\phi})_{A'}{}^{B'} \xi_{(A|B'|}\phi_{B)C}
 -  \tfrac{1}{6} \xi^{CA'} \bar{\phi}_{A'}{}^{B'} (\sExt_{2,0}\phi)_{(A|C|}\xi_{B)B'}\nonumber\\
&- \tfrac{1}{4} \bar{\kappa}^{B'C'} \bar{\phi}_{B'C'} \eta_{(A}{}^{A'}\xi_{B)A'}
 -  \tfrac{1}{12} \kappa^{CD} \phi_{CD} \bar{\eta}^{A'}{}_{(A}\xi_{B)A'}
 -  \tfrac{3}{8} \xi^{CA'} \bar{\eta}_{A'(A}(\sExt_{2,0}\phi)_{BC)}.
\end{align}
\end{lemma}

\section{Conservation laws for the Teukolsky system} \label{sec:teuk} 

In this section we will analyze the tensor $V_{ab}$ defined by \eqref{eq:Vdef} and show that in a Petrov type $\PetrovD$ spacetime with aligned matter condition it is conserved, and depends only on the extreme components $\phi_0, \phi_2$ of the Maxwell field. 

Recall that the operators $\sCurl$ and $\sCurlDagger$ are adjoints, and hence their composition yields a wave operator.  We have the identities (valid in a general spacetime)
\begin{subequations}\label{eq:wavegen}
\begin{align}\label{eq:MaxWaveGen}
\square \varphi_{AB} + 8 \Lambda \varphi_{AB} - 2 \Psi_{ABCD} \varphi^{CD}={}&-2 (\sCurl_{1,1} \sCurlDagger_{2,0} \varphi)_{AB}, \\
\label{eq:PsiWaveGen} 
\square \varphi_{ABCD} - 6 \Psi_{(AB}{}^{FH}\varphi_{CD)FH}={}&-2 (\sCurl_{3,1} \sCurlDagger_{4,0} \varphi)_{ABCD}.
\end{align} 
\end{subequations}
Here $\varphi_{AB}$ and $\varphi_{ABCD}$ are elements of $\SymSpinSec_{2,0}$ and $\SymSpinSec_{4,0}$, respectively. This means that the the Maxwell equation 
$
(\sCurlDagger_{2,0} \phi)_{AA'} = 0
$
in a vacuum spacetime implies the wave equation 
\begin{align}\label{eq:MaxWave}
\square \phi_{AB} - 2 \Psi_{ABCD} \phi^{CD}= 0 .
\end{align} 
Similarly, in a vacuum spacetime, the Bianchi system $(\sCurlDagger_{4,0} \Psi)_{A'ABC} = 0$ holds for the Weyl spinor, and we arrive at the Penrose wave equation 
\begin{align}\label{eq:PsiWave} 
\square \Psi_{ABCD} - 6 \Psi_{(AB}{}^{FH}\Psi_{CD)FH}= 0
\end{align} 
Restricting to a vacuum type $\PetrovD$ spacetime, and projecting the Maxwell wave equation \eqref{eq:MaxWave} and the linearized Penrose wave equation \eqref{eq:PsiWave} on the principal spin dyad, one obtains wave equations for the extreme Maxwell scalars $\phi_0, \phi_2$ and the extreme linearized Weyl scalars $\dot \Psi_0, \dot \Psi_4$. 

\newcommand{\sfrak}{\mathfrak{s}}
Letting $\psi^{(\sfrak)}$ denote $\phi_0, \Psi_2^{-2/3} \phi_2$ for $\sfrak = 1, -1$, respectively, and $\dot \Psi_0, \Psi_2^{-4/3} \dot \Psi_4$ for $\sfrak = 2, -2$, respectively, one finds that these fields satisfy the system  
\begin{equation}\label{eq:squareTME} 
[ \squareTME_{2\sfrak} - 4 \sfrak^2 \Psi_2 ] \psi^{(\sfrak)} = 0 ,
\end{equation} 
see \cite[\S 3]{aksteiner:andersson:2011CQGra..28f5001A}, 
where, in GHP notation 
\begin{equation}\label{eq:squaretme-restrict}
\squareTME_{p} =
2(\tho -p\rho -\bar{\rho})({\tho}'-\rho')- 2(\edt-p\tau
-\bar{\tau}')({\edt}'-\tau')  + (3p-2)\Psi_2 .
\end{equation}
The equation \eqref{eq:squareTME} was first derived by Teukolsky \cite{teukolsky:1972PhRvL..29.1114T,teukolsky:1973} for massless spin-$s$ fields and linearized gravity on Kerr, and is referred to as the Teukolsky Master Equation (TME). It was shown by Ryan \cite{ryan:1974PhRvD..10.1736R} that the tetrad projection of the linearized Penrose wave equation yields the TME, see also Bini et al \cite{bini:etal:2002PThPh.107..967B,bini:etal:2003IJMPD..12.1363B}. In the Kerr case, the TME admits a commuting symmetry operator, and hence allows separation of variables. The TME applies to fields of all half-integer spins between $0$ and $2$. 

As discussed above, the TME is a wave equation for the weighted field $\psi^{(\sfrak)}$.  
It is derived from the spin-$s$ field equation by applying a first order operator and hence is valid for the extreme scalar components of the field, rescaled as explained above. It is important to emphasize that there is a loss of information in deriving the TME from the spin-$s$ equation. For example, if we consider two independent solutions of the TME with spin weights $\sfrak=\pm 1$, these will not in general be components of a single Maxwell field. If indeed this is the case, the Teukolsky-Starobinsky identities (TSI) (also referred to as Teukolsky-Press relations), see \cite{kalnins:etal:1989JMP....30.2925K} and references therein, hold. 

Although the TSI are usually discussed in terms of separated forms of $\psi^{(\sfrak)}$, we are here interested in the TSI as differential relations between the scalars extreme spin weights. From this point of view, the TSI expresses the fact that the Debye potential construction starting from the different Maxwell scalars for a given Maxwell field $\phi_{AB}$ yields scalars of the \emph{the same} Maxwell field. The equations for the Maxwell scalars in terms of Debye potentials can be found in Newman-Penrose notation in \cite{cohen:kegeles:1974PhRvD..10.1070C}. These expressions correspond to the components of a symmetry operator of the second kind. See \cite[\S 5.4.2]{aksteiner:thesis} for further discussion, where also the GHP version of the formulas can be found. 
An analogous situation obtains for the case of linearized gravity, see \cite{lousto:whiting:2002PhRvD..66b4026L}. In this case, the TSI are of fourth order.  
Thus, for a Maxwell field, or a solution of the linearized Einstein equations on a Kerr, or more generally a vacuum type $\PetrovD$ background, the pair of Newman-Penrose scalars of extreme spin weights for the field satisfy a system of differential equations consisting of both the TME and the TSI. 

Although the TME is derived from an equation governed by a variational principle, it has been argued by Anco, see the discussion in \cite{perjes:lukacs_2005AIPC..767..306P}, that the Teukolsky system admits no \emph{real} variational principle, due to the fact that the operator $\squareTME_p $ defined by the above fails to be formally self-adjoint. Hence, the issue of real conserved currents for the Teukolsky system, which appear to be necessary for estimates of the solutions, appears to be open. 
However, as we shall demonstrate here, if we consider the \emph{combined} TME and TSI in the spin-$1$ or Maxwell case, as a system of equations for both of the extreme Maxwell scalars $\phi_0, \phi_2$, this system does admit both a conserved current and a conserved stress-energy like tensor.

\subsection{A new conserved tensor for Maxwell} \label{sec:Vab} 
We have seen in the last section that polarized stress current $-\Psi_{S AA'}[\xi_{AA'}, \chi_{AB}]$ with $\xi^a$ given by \eqref{eq:xikappadef} and $\phi_{AB} \to \chi_{AB}$ the second order symmetry operator of the first kind given by \eqref{eq:SymFirstPot} with $Q = 0$ and $A_{AA'}$ given by \eqref{eq:Adef},  is equivalent to a current $V_{ab} \xi^b$ defined in terms of the symmetric tensor $V_{ab}$.
In fact, as we shall now show, $V_{ab}$ is itself conserved,  
$$
\nabla^a V_{ab} = 0, 
$$
and hence may be viewed as a higher-order stress-energy tensor for the Maxwell field. 
The tensor $V_{ab}$ has several important properties. First of all, it depends only on the extreme Maxwell scalars $\phi_0, \phi_2$, and hence cancels the static Coulomb Maxwell field \eqref{eq:coulomb} on Kerr which has only the middle scalar non-vanishing. 
In order to analyze $V_{ab}$, we first collect some properties of the one-form $\eta_{AA'}$ as defined in \eqref{eq:etadef}. 

\begin{lemma}[\protect{\cite[Lemma 2.4]{2014arXiv1412.2960A}}] 
\label{lem:etaeqs} 
Let $\kappa_{AB} \in \KillSpin_{2,0}$, and assume the aligned matter condition holds with respect to $\kappa_{AB}$. Let $\xi_{AA'}$ be given by \eqref{eq:xikappadef}. Further, let $\phi_{AB}$ be a Maxwell field, and let $\eta_{AA'}$ be given by \eqref{eq:etadef}. Then we have 

\begin{subequations} \label{eq:etafacts} 
\begin{align}
(\sDiv_{1,1} \eta)={}&0,\label{diveta1}\\
(\sCurl_{1,1} \eta)_{AB}={}&\tfrac{2}{3} (\hat{\mathcal{L}}_{\xi}\phi)_{AB},\label{curleta1b2}\\
(\sCurlDagger_{1,1} \eta)_{A'B'}={}&0,  \label{curleta2}\\
\eta_{AA'} \xi^{AA'}={}&\kappa^{AB} (\hat{\mathcal{L}}_{\xi}\phi)_{AB}. \label{eq:etaLphi} 
\end{align}
\end{subequations}
\end{lemma}

\begin{corollary} Assume $\Mcal$ is of Petrov type $\PetrovD$. Then $V_{ab}$ depends only on the extreme components of $\phi_{AB}$. 
\end{corollary}
\begin{proof} 
We first note that by Lemma \ref{lem:Ephi-extreme}, if $\Mcal$ is of type $\PetrovD$, then $(\sExt_{2,0}\phi)_{AB}$ depends only on the extreme components of $\phi_{AB}$, and hence the same property holds for 
$\eta_{AA'}$. Next, recall that if $\Mcal$ is of Petrov type $\PetrovD$, $\kappa_{AB}$ is of algebraic type $\{1,1\}$ and hence $\kappa_{AB}\kappa^{AB} \ne 0$ provided $\kappa_{AB}$ is nonzero. 
A calculation using \eqref{eq:etaLphi} and commutation of $\hat{\mathcal{L}}_{\xi}$ and $\sExt_{2,0}$ now gives 
\begin{align}
(\hat{\mathcal{L}}_{\xi}\phi)_{AB}={}&\frac{\eta^{FF'} \kappa_{AB} \xi_{FF'}}{(\kappa_{AB} \kappa^{AB})}
 -  \frac{(\sExt_{2,0}\hat{\mathcal{L}}_{\xi}\sExt_{2,0}\phi)_{AB}}{2 (\kappa_{AB} \kappa^{AB})}.
 \label{eq:Lxiphieta}
\end{align}
This completes the proof. 
\end{proof} 

\begin{lemma} \label{lem:Tconserved}
Assume that $\varphi_{AB}\in \SymSpinSec_{2,0}$ satisfies the system
\begin{subequations} \label{eq:varphieqs} 
\begin{align}
(\sCurlDagger_{1,1} \sCurlDagger_{2,0} \varphi)_{A'B'}={}&0,\label{eq:CurlDgCurlDgvarphi1}\\
(\sCurl_{1,1} \sCurlDagger_{2,0} \varphi)_{AB}={}&\varpi_{AB},\label{eq:CurlCurlDgvarphi1}
\end{align}
\end{subequations}
for some $\varpi_{AB} \in \SymSpinSec_{2,0}$. 
Let   
 \begin{align}\label{eq:etavarphidef} 
\etavarphi_{AA'}={}&(\sCurlDagger_{2,0} \varphi)_{AA'},
\end{align}
and define the symmetric tensor 
$\EMTensorT_{ABA'B'}$ by 
\begin{align}
\EMTensorT_{ABA'B'}={}&\tfrac{1}{2} \etavarphi_{AB'} \bar{\etavarphi}_{A'B}
 + \tfrac{1}{2} \etavarphi_{BA'} \bar{\etavarphi}_{B'A}
 + \tfrac{1}{2} \bar{\varpi}_{A'B'} \varphi_{AB}
 + \tfrac{1}{2} \varpi_{AB} \bar{\varphi}_{A'B'}.
\end{align}
Then
\begin{align} \label{eq:EMTensorcons} 
\nabla^{BB'}\EMTensorT_{ABA'B'}={}&0.
\end{align}
\end{lemma}
\begin{proof}
By applying the operator $\sCurlDagger_{2,0}$ to \eqref{eq:CurlCurlDgvarphi1}, commuting derivatives and using \eqref{eq:CurlDgCurlDgvarphi1}, we get the integrability condition $(\sCurlDagger_{2,0} \varpi)_{AA'} = 0$. With $\etavarphi_{AA'}$ given by \eqref{eq:etavarphidef}, 
we directly get
\begin{align}
(\sDiv_{1,1} \etavarphi)={}&0,&
(\sCurlDagger_{1,1} \etavarphi)_{A'B'}={}&0,&
(\sCurl_{1,1} \etavarphi)_{AB}={}&\varpi_{AB}.
\end{align}
The equations above give \eqref{eq:EMTensorcons}. 
\end{proof} 
\begin{remark} \label{rem:T}
\begin{enumerate}
\item No assumptions were made on the spacetime geometry in Lemma \ref{lem:Tconserved}.
\item \label{point:DEC} The tensor 
$$
U_{AA'BB'} = \tfrac{1}{2} \etavarphi_{AB'} \bar{\etavarphi}_{A'B}
 + \tfrac{1}{2} \etavarphi_{BA'} \bar{\etavarphi}_{B'A}
 $$
is a super-energy tensor for the 1-form field $\etavarphi_{AA'}$, and hence satisfies the dominant energy condition, cf. \cite{bergqvist:1999CMaPh.207..467B,senovilla:2000CQGra..17.2799S}. In particular, with $\etavarphi_{AA'} = \nabla_{AA'} \psi$ for some scalar $\psi$, then $U_{AA'BB'}$ is just the standard stress-energy tensor for the scalar wave equation, 
$$
U_{ab} =  \nabla_{(a} \psi \nabla_{b)} \bar \psi - \half \nabla^c \psi \nabla_c \bar \psi \met_{ab} 
$$
\item Similarly to the wave equation stress energy, $V_{ab}$ has non-vanishing trace, $V^a{}_a = U^a{}_a = - \bar \etavarphi^a \etavarphi_a$.  
\end{enumerate} 
\end{remark} 

We now have the following result. 
\begin{theorem}[\protect{\cite[Theorem 1.1]{2014arXiv1412.2960A}}] Assume that $(\Mcal, \met_{ab})$ admits a valence $(2,0)$ Killing spinor $\kappa_{AB}$ and assume that the aligned matter condition holds with respect to $\kappa_{AB}$.  
Let $\phi_{AB}$ be a solution of the Maxwell equation. Then the tensor 
$V_{ABA'B'}$ given by \eqref{eq:Vdef} is conserved, i.e. 
$$
\nabla^{AA'} V_{ABA'B'} = 0
$$
If in addition $(\Mcal, \met_{ab})$ is of Petrov type $\PetrovD$, then $V_{ab}$ depends only on the extreme components of $\phi_{AB}$.  
\end{theorem} 

\begin{remark} 
\begin{enumerate} 
\item A first order conserved tensor for the Maxwell field has previously been found by  Bergquist et al. \cite{bergqvist:etal:2003CQGra..20.2663B}. As they showed,  the tensor $B_{ab} = \nabla^c \phi_{AB} \nabla_c \bar \phi_{A'B'}$ is conserved in a Ricci flat spacetime. However, one may demonstrate that in the Kerr spacetime, the current $B_{ab} \xi^b$ is trivial in the sense of Definition \ref{def:trivcurr}.  
On the other hand the current $V_{ab} \xi^b$ is non-trivial in the Kerr spacetime. See \cite{ABB:currents}. 
\item The correction terms in $V_{ab}$ involving $(\sExt_{2,0}\phi)_{AB} (\hat{\mathcal{L}}_{\xi}\bar{\phi})_{A'B'}$ and its complex conjugate are of first order in derivatives of the Maxwell field, which opens the possibility of dominating these using a Cauchy-Schwarz argument involving $V_{ab}$ and the zeroth order Maxwell stress energy $T_{ab} = \phi_{AB} \bar \phi_{A'B'}$. 
\end{enumerate} 
\end{remark}  
The properties of $V_{ab}$ discussed above indicate that $V_{ab}$, rather than the Maxwell stress-energy $T_{ab}$ may be used in proving dispersive estimates for the Maxwell field. In fact, it is immediately clear that the Maxwell stress energy cannot be used directly to prove dispersive estimates since it does not vanish for the Coulomb field on the Kerr spacetime. 

\subsection{Teukolsky equation and conservation laws} 
\begin{lemma} 
Assume that $(\Mcal, \met_{ab})$ is a type D spacetime which admits a valence $(2,0)$ Killing spinor $\kappa_{AB}$ and assume that the aligned matter condition holds with respect to $\kappa_{AB}$. Let $\phi_{AB}$ be a solution of the Maxwell equation. Then $\phi_{AB}$ is a solution of the system of equations 
\begin{subequations} \label{eq:TME-TSI-cov}
\begin{align} 
(\sCurlDagger_{1,1} \sCurlDagger_{2,0} \sExt_{2,0} \phi)_{A'B'} &= 0 
\label{eq:bCdCdEphi}
\\
(\sExt_{2,0}\sCurl_{1,1} \sCurlDagger_{2,0} \sExt_{2,0} \phi)_{AB}={}&\tfrac{2}{3} (\hat{\mathcal{L}}_{\xi}\sExt_{2,0} \phi)_{AB}.
\label{eq:ECCdEphi}
\end{align} 
\end{subequations}
Furthermore, this system is equivalent to \eqref{eq:CurlDgCurlDgvarphi1}, \eqref{eq:CurlCurlDgvarphi1}, with $\varphi_{AB}=\sExt_{2,0} \phi$ and $\varpi_{AB}=\tfrac{2}{3} (\hat{\mathcal{L}}_{\xi}\phi)_{AB}$.
\end{lemma}
\begin{proof} The equations \eqref{eq:bCdCdEphi} and \eqref{eq:ECCdEphi} follows from equations \eqref{curleta1b2}, \eqref{curleta2} and the fact that $\sExt_{2,0}$ and $\hat{\mathcal{L}}_{\xi}$ commutes. The equation \eqref{curleta1b2}, can be split into two parts, one where $\sExt_{2,0}$ is applied, and the other when the equation is contracted with $\kappa^{AB}$. The latter part can be seen by differentiating $\kappa^{AB} (\mathcal{E}_{2,0}\phi)_{AB} = 0$ twice. 
\end{proof} 
In order to compare equations \eqref{eq:TME-TSI-cov} with the scalar forms of the TME and TSI, we now project these equations on the dyad. A calculation gives the following result. 
\begin{lemma} \label{lem:TSIGHP} 
\begin{enumerate}
\item 
The GHP form of equation \eqref{eq:bCdCdEphi} is 
\begin{subequations}
\begin{align}
0={}&-2 \rho \tho \varphi_2
 + \tho \tho \varphi_2
 - 2 \tau ' \edt' \varphi_0
 + \edt' \edt' \varphi_0,\\
0={}&- \tau \tho \varphi_2
 + \tfrac{1}{2} \overline{\tau '} \tho \varphi_2
 + \tfrac{1}{2} \tho \edt \varphi_2
 + \tfrac{1}{2} \bar{\tau} \tho' \varphi_0
 -  \tau ' \tho' \varphi_0
 + \tfrac{1}{2} \tho' \edt' \varphi_0
 -  \rho \edt \varphi_2\nonumber\\
& + \tfrac{1}{2} \bar{\rho} \edt \varphi_2
 + \tfrac{1}{2} \edt \tho \varphi_2
 -  \rho ' \edt' \varphi_0
 + \tfrac{1}{2} \overline{\rho '} \edt' \varphi_0
 + \tfrac{1}{2} \edt' \tho' \varphi_0,\\
0={}&-2 \rho ' \tho' \varphi_0
 + \tho' \tho' \varphi_0
 - 2 \tau \edt \varphi_2
 + \edt \edt \varphi_2,
\end{align}
\end{subequations}
where $\varphi_{0}=-2 \kappa_1 \phi_0$ and $\varphi_{2}=2 \kappa_1 \phi_2$. 
\item 
The GHP form of equation \eqref{eq:ECCdEphi} is
\begin{subequations}
\begin{align}
0={}&- \tho \tho' \varphi_0
 + \rho \tho' \varphi_0
 + \bar{\rho} \tho' \varphi_0
 + \edt \edt' \varphi_0
 -  \tau \edt' \varphi_0
 -  \overline{\tau '} \edt' \varphi_0,\\
0={}&- \rho ' \tho \varphi_2
 -  \overline{\rho '} \tho \varphi_2
 + \tho' \tho \varphi_2
 + \bar{\tau} \edt \varphi_2
 + \tau ' \edt \varphi_2
 -  \edt' \edt \varphi_2.
\end{align}
\end{subequations}
\end{enumerate}
\end{lemma} 
\begin{remark} We see from \ref{lem:TSIGHP} that 
equation \eqref{eq:bCdCdEphi} is equivalent to the TSI for Maxwell given in scalar form in \cite[\S 5.4.2]{aksteiner:thesis}, while equation \eqref{eq:ECCdEphi} is equivalent to the scalar form of TME for Maxwell given in \eqref{eq:squareTME} above. 
\end{remark}

\section{A Morawetz estimate for the Maxwell field on Schwarzschild} \label{sec:morawetz} 

As discussed in section \ref{sec:geodesics}, one may construct a suitable function of the conserved quantities for null geodesics in the Kerr spacetime which is monotone along the geodesic flow. This function may be viewed as arising from a generalized vector field on phase space. The monotonicity property implies, as discussed there, that non-trapped null geodesics disperse, in the sense that they leave any stationary region in the Kerr space time.  As mentioned in section \ref{sec:geodesics}, in view of the geometric optics approximation for the wave equation, such a monotonicity property for null geodesics reflects the tendency for waves in the Kerr spacetime to disperse. 

At the level of the wave equation, the analogue of the just mentioned monotonicity estimate is called the Morawetz estimate. For the wave equation $\nabla^a \nabla_a \psi = 0$, a Morawetz estimate provides a current $J_a$ defined in terms of $\psi$ and some of its derivatives, with the property that $\nabla^a J_a$ has suitable positivity properties, and that the flux of $J_a$ can be controlled by a suitable energy defined in terms of the field. 

Let $\psi$ be a solution of the wave equation $\nabla^a \nabla_a \psi = 0$. 
The stress-energy tensor $T_{ab}$ for $\psi$ is 
$$
T_{ab} = \nabla_{(a} \psi \nabla_{b)} \bar \psi - \half \nabla^c \psi \nabla_c \bar \psi g_{ab}
$$

Define the current $J_a$ by 
$$
J_a = T_{ab} \vecMprimary^b +  \half \scalMprimary(\bar \psi \nabla_a \psi  + \psi \nabla_a \bar\psi)  - \half \nabla_a \scalMprimary \psi \bar \psi.
$$
We have 
\begin{equation}\label{eq:JJJJ}
\nabla^a J_a = T_{ab} \nabla^{(a} \vecMprimary^{b)} + \scalMprimary \nabla^c \psi \nabla_c \bar \psi - \half (\nabla^c \nabla_c \scalMprimary) \psi \bar \psi .
\end{equation}
We now specialize to Minkowski space, with the line element 
$g_{ab}dx^a dx^b = dt^2 - dr^2 - d\theta^2 - r^2 \sin^2\theta d\phi^2$. 
Let 
$$
E(\tau) = \int_{\{ t= \tau\}} T_{tt} d^3 x
$$
be the energy of the field at time $\tau$, where $T_{tt}$ is the energy density. The energy is conserved, so that $E(t)$ is independent of $t$. 

Setting $\vecMprimary^a = r(\partial_r)^a$, we have 
\begin{equation} \label{eq:deform} 
\nabla^{(a} \vecMprimary^{b)} = g^{ab} - (\partial_t)^a (\partial_t)^b . 
\end{equation} 
With $q = 1$,  we get
$$
\nabla^a J_a = - T_{tt} .
$$
With the above choices, the bulk term $\nabla^a  J_a$ has a sign. This method can be used to prove dispersion for solutions of the wave equation. In particular, by introducing suitable cutoffs, one finds that for any $R_0 > 0$, there is a constant $C$, so that  
\begin{align}
\int_{t_0}^{t_1} \int_{|r| \leq R_0}  T_{tt} d^3 x dt 
\leq C(E(t_0) +E(t_1)) 
\leq 2C E(t_0) , \label{eq:basicmor} 
\end{align}
see \cite{MR0234136}. 
The local energy, $\int_{|r| \leq R_0} T_{tt} d^3 x$, is a function of
time. By \eqref{eq:basicmor} it is integrable in $t$, and hence it must decay to zero
as $t\rightarrow\infty$, at least sequentially. This shows that the field disperses. 
Estimates of this type are called
Morawetz or integrated local energy decay estimates. 

For the Maxwell field on Minkowski space, we have 
$$
T_{ab} = \phi_{AB} \bar \phi_{A'B'}
$$
 with $T^a{}_a = 0$. Setting $J_a = T_{ab} \vecMprimary^b$, with $\vecMprimary^a = r(\partial_r)^a$ we have
$$
\nabla^a J_a = - T_{tt}
$$
which again gives local energy decay for the Maxwell field on Minkowski space.

For the wave equation on Schwarzschild we can choose  
\begin{subequations}
\begin{align}
\vecMprimary^{a}={}&\frac{(r - 3 M) (r - 2 M)}{3 r^2} (\partial_r)^{a} ,\\
\scalMprimary={}&\frac{6 M^2 - 7 M r + 2 r^2}{6 r^3}.
\end{align}
\end{subequations}
This gives

\begin{subequations}
\begin{align}\label{eq:Afirst}
-\nabla^{(a}\vecMprimary^{b)}={}&- \frac{M g^{ab} (r - 3 M)}{3 r^3}
 +  \frac{M (r - 2 M)^2 (\partial_r)^{a} (\partial_r)^{b}}{r^4}\nonumber\\
& +  \frac{(r - 3 M)^2 ((\partial_\theta)^{a} (\partial_\theta)^{b} + \csc^2\theta (\partial_\phi)^{a} (\partial_\phi)^{b})}{3 r^5},\\
- \nabla_{a}J^{a}={}&\frac{M |\partial_r \psi|^2 (r - 2 M)^2}{r^4}
 + \frac{\bigl(|\partial_\theta \psi|^2 + |\partial_\phi\psi|^2 \csc^2\theta\bigr) (r - 3 M)^2}{3 r^5}\nonumber\\
& + \frac{M |\psi|^2 (54 M^2 - 46 M r + 9 r^2)}{6 r^6}. \label{eq:divJ}
\end{align}
\end{subequations} 
Here, $\vecMprimary^a$ was chosen so that the last two terms \eqref{eq:Afirst} have good signs. The form of $\scalMprimary$ given here was chosen to eliminate the $|\partial_t \psi|^2$ term in \eqref{eq:divJ}. 
The first terms in \eqref{eq:divJ} are clearly non-negative, while the last is of lower-order and can be estimated using a Hardy estimate \cite{andersson:blue:0908.2265}.
The effect of trapping in Schwarzschild at $r=3M$ is manifested in the fact that the angular derivative term 
vanishes at $r=3M$. 

In the case of the wave equation on Kerr, the above argument using a classical vector field cannot work due to the complicated structure of the trapping. However, making use of higher-order currents constructed using second order symmetry operators for the wave equation, and a generalized Morawetz vector field analogous to the vector field $\vecMprimary^a$ as discussed in section \ref{sec:geodesics}. This approach has been carried out in detail in \cite{andersson:blue:0908.2265}. 

If we apply the same idea for the Maxwell field on Schwarzschild, there is no reason to expect that local energy decay should hold, in view of the fact that the Coulomb solution is a time-independent solution of the Maxwell equation which does not disperse. In fact, with 
\begin{align}
\vecMprimary^{a}={}& \fnMrMorawetz(r) \Bigl(1
 -  \frac{2 M}{r} \Bigr) (\partial_r)^{a},\\
- \phi^{AB} \bar{\phi}^{A'B'} (\sTwist_{1,1} P)_{ABA'B'}={}&\bigl(|\phi_{00}|^2 + |\phi_{11}|^2\bigr) \frac{(r - 2 M)}{2 r} \fnMrMorawetz'(r) \nonumber\\
& - \frac{|\phi_{01}|^2 \bigl(r (r-2M) \fnMrMorawetz'(r) -  2 \fnMrMorawetz(r) (r - 3 M) \bigr)}{r^2}.
\label{eq:ssss} 
\end{align}
If $\fnMrMorawetz'$ is chosen to be positive, then the coefficient of the extreme components in \eqref{eq:ssss} is positive. However, at $r=3M$, the coefficient of the middle component is necessarily of the opposite sign. It is possible to show that no choice of $\fnMrMorawetz$ will give positive coefficients for all components in \eqref{eq:ssss}. 

The dominant energy condition, that  $T_{ab} V^a W^b \geq 0$ for all causal vectors  $V^a, W^a$
is a common and important condition on stress energy tensors. 
In Riemannian geometry, a natural condition on a
symmetric $2$-tensor $T_{ab}$ would be non-negativity, i.e. the condition that
for all $X^a$, one has $T_{ab}X^a X^b\geq
0$. 

However, in order to prove dispersive estimates for null geodesics and the wave equation,
the dominant energy condition on its own is not sufficient and non-negativity cannot be expected for stress energy tensors. Instead, a useful condition to consider is non-negativity modulo trace
terms, i.e. the condition that for every $X^a$ there is a $q$ such that
$T_{ab} X^a X^b +q T^a{}_a \geq 0$. For null
geodesics and the wave equation, the tensors $\dot\gamma_a\dot\gamma_b$
and $\nabla_a u\nabla_b u=T_{ab} +T^\gamma{}_\gamma
g_{ab}$ are both non-negative, so
$\dot\gamma_a\dot\gamma_b$ and $T_{ab}$ are
non-negative modulo trace terms. 

From equation \eqref{eq:Afirst}, we see
that $-\nabla^{(a}A^{b)}$ is of the form $f_1 g^{ab}
+f_2\partial_r^a\partial_r^b
+f_3\partial_\theta^a\partial_\theta^b+f_4\partial_\phi^a\partial_\phi^b$ where $f_2$, $f_3$ and $f_4$ are
non-negative functions. That is $-\nabla^{(a}A^{b)}$ is a sum
of a multiple of the metric plus a sum of terms of the form of a
non-negative coefficient times a vector tensored with itself. Thus,
from the non-negativity modulo trace terms, for null geodesics and the
wave equation respectively, there are functions $q$ such that
$\dot\gamma_a\dot\gamma_b \nabla^a
A^b=\dot\gamma_a\dot\gamma_b \nabla^a A^b +q
g^{ab}\dot\gamma_a\dot\gamma_b \leq 0$ and
$T_{ab}\nabla^a A^b +qT^a{}_a\leq 0$. For
null geodesics, since
$g^{ab}\dot\gamma_a\dot\gamma_b=0$, the $q$ term can
be ignored. For the wave equation, one can use the terms involving $\scalMprimary$
in equations \eqref{eq:JJJJ}, to cancel the
$T^a{}_a$
term in $\nabla^a J_a$. For the wave equation, this gives
non-negativity for the first-order terms in $-\nabla^a J_a$,
and one can then hope to use a Hardy estimate to control the zeroth
order terms.

If we now consider the Maxwell equation, we have the fact that the Maxwell stress energy tensor is traceless, $T^a{}_a = 0$ and does not satisfy
the non-negativity condition. Therefore it also does not satisfy the condition of non-negativity modulo trace. This appears to be the
fundamental underlying obstruction to proving a Morawetz estimate
using $T_{ab}$. This can be seen as a manifestation of the fact that the Coloumb solution does not disperse. 

Attempts to overcome this problem were a
major motivating factor in our efforts to understand  conserved
currents and tensors for the Maxwell equation, other than the ones derived from $T_{ab}$. As we shall see, an important observation is that the tensor  
$$
U_{AA'BB'} = \tfrac{1}{2} \eta_{AB'} \bar{\eta}_{A'B}
 + \tfrac{1}{2} \eta_{BA'} \bar{\eta}_{B'A}
 $$
which is the leading order part of the higher-order conserved tensor $V_{ab}$ satisfies the non-negativity modulo trace terms condition.

\subsection{Morawetz for Maxwell using $V_{ab}$} 

In this section we shall apply the first order stress energy tensor $V_{ab}$ given in \eqref{eq:Vdef} to construct a Morawetz estimate for the Maxwell field on the Schwarzschild spacetime. In order to do this, we shall use a radial Morawetz vector field $\vecMprimary^{AA'}$ as explained above for the wave equation, together with lower-order correction terms with a scalar field $\scalMprimary$ analogous to the one used there. However, due to the cross terms in $V_{ab}$ involving $(\sExt_{2,0} \phi)_{AB} \Lie_\xi \overline{\phi}_{A'B'}$ and its complex conjugate, additional correction terms are needed. These are given below and involve the vector field $\vecMsecondary^{AA'}$ and the scalar function $\scalMsecondary$. 

Define  
\begin{subequations}
\begin{align}
J^{AA'}={}&\vecMprimary^{BB'} V^{A}{}_{B}{}^{A'}{}_{B'} \nonumber \\ 
& + \tfrac{1}{2} \scalMprimary \bar{\eta}^{A'}{}_{B} (\sExt_{2,0}\phi)^{AB}
 + \tfrac{1}{2} \scalMprimary \eta^{A}{}_{B'} (\bar\sExt_{0,2}\bar\phi)^{A'B'} 
 + \tfrac{1}{2} (\sExt_{2,0}\phi)^{AB} (\bar\sExt_{0,2}\bar\phi)^{A'B'} (\sTwist_{0,0} \scalMprimary)_{BB'}
 \nonumber \\
&+ \vecMsecondary_{BB'} (\sExt_{2,0}\phi)^{AB} (\bar\sExt_{0,2}\bar\phi)^{A'B'} 
+ \scalMsecondary W \xi^{AA'} ,\\
\intertext{where} 
W={}&((\sExt_{2,0}\phi)_{AC} (\bar\sExt_{0,2}\bar\phi)_{A'B'} \kappa_{B}{}^{C}
 + (\sExt_{2,0}\phi)_{AB} (\bar\sExt_{0,2}\bar\phi)_{A'C'} \bar{\kappa}_{B'}{}^{C'}) \xi^{AA'} \xi^{BB'}.
 \label{eq:Wdef} 
\end{align}
\end{subequations}
From this, we get 
\begin{align}
\xi^{AA'} (\sTwist_{0,0} W)_{AA'}={}&O_{ABA'B'} (\bar\sExt_{0,2}\bar\phi)^{A'B'} \Xi^{AB}
 + \overline{O}_{A'B'AB} (\sExt_{2,0}\phi)^{AB} \overline{\Xi}^{A'B'},
\end{align}
where
\begin{subequations}
\begin{align}
\Xi_{AB}={}&\tfrac{1}{2} \mathcal{L}_{\xi}\phi_{AB}
 -  \frac{\mathcal{L}_{\xi}\phi^{CD} \kappa_{AC} \kappa_{BD}}{(\kappa_{AB} \kappa^{AB})},\\
O_{AB}{}^{A'B'}={}&(\kappa_{AB} \kappa^{AB}) \xi_{(A}{}^{(A'}\xi_{B)}{}^{B')}
 + 2 \kappa_{A}{}^{C} \bar{\kappa}^{A'}{}_{C'} \xi_{(B}{}^{(B'}\xi_{C)}{}^{C')}.
\end{align}
\end{subequations}
We remark that $\Xi_{AB}$ contains only the extreme components of $\Lie_\xi \phi_{AB}$ and has vanishing middle component.   
The divergence of the current, $\nabla^a J_a = \sDiv_{1,1} J$, is given by
\begin{align}
- (\sDiv_{1,1} J)={}& -  V^{ABA'B'} (\sTwist_{1,1} \vecMprimary)_{ABA'B'} + \tfrac{1}{4} \eta^{BB'} \bar{\eta}_{B'B} (\sDiv_{1,1} \vecMprimary) \nonumber - \scalMprimary \eta^{BB'} \bar{\eta}_{B'B} \nonumber \\
&
 -  \vecMsecondary^{BB'} \bar{\eta}_{B'}{}^{A} (\sExt_{2,0}\phi)_{BA}
 -  \vecMsecondary^{BB'} \eta_{B}{}^{A'} (\bar\sExt_{0,2}\bar\phi)_{B'A'}
 -  \scalMsecondary \xi^{BB'} (\sTwist_{0,0} W)_{BB'}
 \nonumber\\
& 
 -  (\sExt_{2,0}\phi)^{BA} (\bar\sExt_{0,2}\bar\phi)^{B'A'} (\sTwist_{1,1} \vecMsecondary)_{BAB'A'}
 -  \tfrac{1}{2} (\sExt_{2,0}\phi)^{BA} (\bar\sExt_{0,2}\bar\phi)^{B'A'} (\sTwist_{1,1} \sTwist_{0,0} \scalMprimary)_{BAB'A'}
 \nonumber \\ 
 & -  W \xi^{BB'} (\sTwist_{0,0} \scalMsecondary)_{BB'} .
\end{align}
The fields $\vecMprimary^{AA'}$, $\vecMsecondary^{AA'}$, $\scalMprimary$ and $\scalMsecondary$ can then be engineered to obtain a sign for the integrated divergence of $J^{AA'}$. 
The fields $\vecMprimary^{AA'}$ and $\scalMprimary$ are chosen so that the terms involving $\vecMprimary^{AA'}$, $\scalMprimary$ give a non-negative quadratic form in $\eta_{AA'}$. 
Terms analogous to those generated by the cross terms in $V_{ab}$ of the form $(\sExt_{2,0}\phi)_{AB} \Lie_\xi \bar{\phi}_{A'B'}$ are not present in the case of the wave equation. These are controlled using the scalar $\scalMsecondary$. Here it is important to note that the terms involve $\Xi_{AB}$ can be related to those involving $\Lie_\xi \phi_{AB}$ using the identity
$$
 \mathcal{L}_{\xi}\phi_{AB} = \frac{\eta^{CA'} \kappa_{AB} \xi_{CA'}}{(\kappa_{AB} \kappa^{AB})} + \Xi_{AB}.
$$
This is related to the fact that the middle component of $\Lie_\xi \phi_{AB}$ can be written in terms of $\eta_{AA'}$, cf. \eqref{eq:etaLphi}.
The vector $\vecMsecondary^{AA'}$ together with a subsequent modification of $\scalMprimary$ allows us to write a quadratic form not in $\eta_{AA'}$, but in $\eta_{AA'}$ plus lower-order terms. This allows us to modify the remaining term quadratic in $(\sExt_{2,0}\phi)_{AB}$ so it can be estimated with a Hardy estimate on the sphere.

Recall that the Schwarzschild metric is the Kerr metric with $a=0$. Choosing the principal tetrad in Schwarzschild given by specializing \eqref{eq:KerrTetrad} to $a = 0$, gives in a standard manner an orthonormal frame, 
\begin{align*}
T^{AA'}\equiv{}&\tfrac{1}{\sqrt{2}}(o^{A} \bar o^{A'}
 + \iota^{A} \bar\iota^{A'}),&
X^{AA'}\equiv{}&\tfrac{1}{\sqrt{2}}(\bar o^{A'} \iota^{A}
 + o^{A} \bar\iota^{A'}),\\
Y^{AA'}\equiv{}&\tfrac{i}{\sqrt{2}}(- \bar o^{A'} \iota^{A}
 + o^{A} \bar\iota^{A'}),&
Z^{AA'}\equiv{}&\tfrac{1}{\sqrt{2}}(o^{A} \bar o^{A'}
 -  \iota^{A} \bar\iota^{A'}).
\end{align*}
Choosing the vector fields $\vecMprimary^{AA'}, \vecMsecondary^{AA'}$ to be radial vector fields, with radial coefficients, and the scalars $\scalMprimary, \scalMsecondary$ as radial functions, where  $\scalMsecondary$ is in addition chosen to eliminate cross terms involving $\Lie_\xi \phi_{AB} \overline{\Xi}_{A'B'}$, the divergence $\sDiv_{1,1} J$ can be written in the form 
\begin{align}
- (\sDiv_{1,1} J)={}&\zeta |T^{AA'} \eta_{AA'}|^{2}
 + \zeta |Z^{AA'} \eta_{AA'}|^{2}
 + E[\eta_{AA'} + \beta Z^{C}{}_{A'} (\sExt_{2,0}\phi)_{AC}]\nonumber\\
& + \varsigma T^{AA'} T^{BB'} (\sExt_{2,0}\phi)_{AB} (\bar{\sExt}_{0,2} \bar{\phi})_{A'B'} ,
\label{eq:DivJMorawetz2}
\end{align}
where $\zeta$, $\beta$ and $\varsigma$ are radial functions completely determined $\vecMprimary^{AA'}, \vecMsecondary^{AA'}, \scalMprimary$, and where $E$ is a quadratic expression in its argument of the form 
\begin{align}
E[\nu_{AA'}]={}&\mathscr{A} \nu^{AA'} \bar{\nu}^{B'B} T_{(A|A'|}T_{B)B'} + \mathscr{B} \nu^{AA'} \bar{\nu}^{B'B} Z_{(A|A'|}Z_{B)B'}  
 + \mathscr{C} \nu^{AA'} \bar{\nu}_{A'A} , \label{eq:Eform}
\end{align}
where the coefficients depend on the choice of $\vecMprimary^{AA'}, \vecMsecondary^{AA'}, \scalMprimary$. We note that the first two terms in \eqref{eq:DivJMorawetz2} are non-negative provided $\zeta$ is non-negative. Further, it is possible to arrange that $E$ given by \eqref{eq:Eform} is non-negative, at the same time as $\zeta$ is non-negative. 

In Schwarzschild, the vector $T^{AA'}$ is timelike outside the horizon, and hence due to the fact that the tensor $(\sExt_{2,0}\phi)_{AB} (\bar{\sExt}_{0,2} \bar{\phi})_{A'B'}$ has the dominant energy condition, the expression $T^{AA'} T^{BB'} (\sExt_{2,0}\phi)_{AB} (\bar{\sExt}_{0,2} \bar{\phi})_{A'B'}$ is positive outside the horizon. However, it is not possible to choose the coefficient function $\varsigma$ to be non-negative everywhere. 
In order to overcome this obstacle, we use a Hardy estimate on the sphere of radius $r$, 
\begin{align}
\int_{S_r}|T^{AA'}\eta_{AA'}|^2 + |Z^{AA'}\eta_{AA'}|^2 \mbox{d}\mu_{S_r}
\geq{}& \frac{2}{r^2}\int_{S_r}T^{AA'}T^{BB'}(\sExt_{2,0}\phi)_{AB}(\bar\sExt_{0,2}\bar\phi)_{A'B'} \mbox{d}\mu_{S_r}.
\end{align}
This estimate, together with the positivity of the first two terms in \eqref{eq:DivJMorawetz2} allows us to to compensate for the negative part in the last term, after integration. 

Putting these ideas together, it is possible to prove a Morawetz estimate for the Maxwell equation on the Schwarzschild spacetime. In the paper \cite{2015arXiv150104641A} we give a complete proof of a Morawetz estimate using a slight modification of $V_{ab}$ which is not conserved, but which simpler and still gives a conserved energy current in the Schwarzschild case.  

If we consider the Maxwell field on the Kerr spacetime, the approach based on $V_{ab}$ outlined above generalizes. However, as in the case of the wave equation on Kerr, one is faced with difficulties caused by the complicated trapping. As for the wave equation, one expects that higher-order currents constructed using the second order symmetry operators for the Maxwell field discussed in this paper, cf. \cite{ABB:symop:2014CQGra..31m5015A}, can be applied along the lines of 
\cite{andersson:blue:0908.2265} to prove a Morawetz estimate for the Maxwell field also in this case, see 
\cite{ABB:futurekerr}.

\subsection*{Acknowledgements} We are grateful to Steffen Aksteiner and Walter Simon for discussions and helpful remarks.  L.A. thanks the Mathematical Sciences Center, Tsinghua University, Beijing, for hospitality during part of this work. T.B. and P.B. acknowledges support from grant EP/J011142/1 from the Engineering and Physical Sciences Research Council. L.A. acknowledges support from the Knut and Alice Wallenberg foundation.

\newcommand{\arxivref}[1]{\href{http://www.arxiv.org/abs/#1}{{arXiv.org:#1}}}
\newcommand{\mnras}{Monthly Notices of the Royal Astronomical Society}
\newcommand{\prd}{Phys. Rev. D} 

\begin{thebibliography}{10}

\bibitem{aksteiner:thesis}
S.~Aksteiner.
\newblock {\em Geometry and analysis in black hole spacetimes}.
\newblock PhD thesis, Gottfried Wilhelm Leibniz Universit\"at Hannover, 2014.

\bibitem{aksteiner:andersson:2011CQGra..28f5001A}
S.~{Aksteiner} and L.~{Andersson}.
\newblock {Linearized gravity and gauge conditions}.
\newblock {\em Classical and Quantum Gravity}, 28(6):065001, Mar. 2011.
\newblock \arxivref{1009.5647}.

\bibitem{aksteiner:andersson:2013CQGra..30o5016A}
S.~{Aksteiner} and L.~{Andersson}.
\newblock {Charges for linearized gravity}.
\newblock {\em Classical and Quantum Gravity}, 30(15):155016, Aug. 2013.
\newblock \arxivref{1301.2674}.

\bibitem{2013arXiv1304.0487A}
S.~Alexakis, A.~D. Ionescu, and S.~Klainerman.
\newblock Rigidity of stationary black holes with small angular momentum on the
  horizon.
\newblock {\em Duke Math. J.}, 163(14):2603--2615, 2014.
\newblock \arxivref{1304.0487}.

\bibitem{anco:pohjanpelto:2001:MR1885280}
S.~C. Anco and J.~Pohjanpelto.
\newblock Classification of local conservation laws of {M}axwell's equations.
\newblock {\em Acta Appl. Math.}, 69(3):285--327, 2001.

\bibitem{anco:pohjanpelto:2003:ProcRSoc:MR1997098}
S.~C. Anco and J.~Pohjanpelto.
\newblock Conserved currents of massless fields of spin {$s\geq\frac 12$}.
\newblock {\em R. Soc. Lond. Proc. Ser. A Math. Phys. Eng. Sci.},
  459(2033):1215--1239, 2003.

\bibitem{anco:pohjanpelto:2003:CRM:MR2056970}
S.~C. Anco and J.~Pohjanpelto.
\newblock Symmetries and currents of massless neutrino fields, electromagnetic
  and graviton fields.
\newblock In {\em Symmetry in physics}, volume~34 of {\em CRM Proc. Lecture
  Notes}, pages 1--12. Amer. Math. Soc., Providence, RI, 2004.
\newblock math-ph/0306072.

\bibitem{anco:the:2005:AAM:MR2220197}
S.~C. Anco and D.~The.
\newblock Symmetries, conservation laws, and cohomology of {M}axwell's
  equations using potentials.
\newblock {\em Acta Appl. Math.}, 89(1-3):1--52 (2006), 2005.

\bibitem{ABB:currents}
L.~{Andersson}, T.~{B{\"a}ckdahl}, and P.~{Blue}.
\newblock {Conserved currents}.
\newblock In preparation.

\bibitem{ABB:futurekerr}
L.~{Andersson}, T.~{B{\"a}ckdahl}, and P.~{Blue}.
\newblock {Decay estimates for the Maxwell field on the Kerr spacetime}.
\newblock In preparation.

\bibitem{2014arXiv1412.2960A}
L.~{Andersson}, T.~{B{\"a}ckdahl}, and P.~{Blue}.
\newblock {A new tensorial conservation law for Maxwell fields on the Kerr
  background}.
\newblock 2014.
\newblock \arxivref{1412.2960}.

\bibitem{ABB:symop:2014CQGra..31m5015A}
L.~{Andersson}, T.~{B{\"a}ckdahl}, and P.~{Blue}.
\newblock {Second order symmetry operators}.
\newblock {\em Classical and Quantum Gravity}, 31(13):135015, July 2014.
\newblock \arxivref{1402.6252}.

\bibitem{2015arXiv150104641A}
L.~{Andersson}, T.~{B{\"a}ckdahl}, and P.~{Blue}.
\newblock {Decay of solutions to the Maxwell equation on the Schwarzschild
  background}.
\newblock Jan. 2015.
\newblock \arxivref{1501.04641}.

\bibitem{andersson:blue:0908.2265}
L.~{Andersson} and P.~{Blue}.
\newblock {Hidden symmetries and decay for the wave equation on the Kerr
  spacetime}.
\newblock Aug. 2009.
\newblock \arxivref{0908.2265}.

\bibitem{andersson:blue:2013arXiv1310.2664A}
L.~{Andersson} and P.~{Blue}.
\newblock {Uniform energy bound and asymptotics for the Maxwell field on a
  slowly rotating Kerr black hole exterior}.
\newblock Oct. 2013.
\newblock \arxivref{1310.2664}.

\bibitem{Bae11a}
T.~{B\"{a}ckdahl}.
\newblock Sym{M}anipulator, 2011-2014.
\newblock
  \href{http://www.xact.es/SymManipulator}{http://www.xact.es/SymManipulator}.

\bibitem{backdahl:valiente-kroon:2012JMP....53d2503B}
T.~{B{\"a}ckdahl} and J.~A.~V. {Kroon}.
\newblock {Constructing ``non-Kerrness'' on compact domains}.
\newblock {\em Journal of Mathematical Physics}, 53(4):042503, Apr. 2012.

\bibitem{backdahl:valiente-kroon:2010PhRvL.104w1102B}
T.~{B{\"a}ckdahl} and J.~A. {Valiente Kroon}.
\newblock {Geometric Invariant Measuring the Deviation from Kerr Data}.
\newblock {\em Physical Review Letters}, 104(23):231102, June 2010.

\bibitem{backdahl:valiente-kroon:2010:MR2753388}
T.~B{\"a}ckdahl and J.~A. Valiente~Kroon.
\newblock On the construction of a geometric invariant measuring the deviation
  from {K}err data.
\newblock {\em Ann. Henri Poincar\'e}, 11(7):1225--1271, 2010.

\bibitem{backdahl:valente-kroon:2011RSPSA.467.1701B}
T.~{B{\"a}ckdahl} and J.~A. {Valiente Kroon}.
\newblock {The 'non-Kerrness' of domains of outer communication of black holes
  and exteriors of stars}.
\newblock {\em Royal Society of London Proceedings Series A}, 467:1701--1718,
  June 2011.
\newblock \arxivref{1010.2421}.

\bibitem{Baer:MR1224089}
C.~B{\"a}r.
\newblock Real {K}illing spinors and holonomy.
\newblock {\em Comm. Math. Phys.}, 154(3):509--521, 1993.

\bibitem{chrusciel:bartnik:2003math......7278C}
R.~A. Bartnik and P.~T. Chru{\'s}ciel.
\newblock Boundary value problems for {D}irac-type equations.
\newblock {\em J. Reine Angew. Math.}, 579:13--73, 2005.
\newblock \arxivref{math/0307278}.

\bibitem{Batista:2012}
C.~{Batista}.
\newblock {Weyl tensor classification in four-dimensional manifolds of all
  signatures}.
\newblock {\em General Relativity and Gravitation}, 45:785--798, Apr. 2013.
\newblock \arxivref{1204.5133}.

\bibitem{beig:chrusciel:1996JMP....37.1939B}
R.~{Beig} and P.~T. {Chru{\'s}ciel}.
\newblock {Killing vectors in asymptotically flat space-times. I.
  Asymptotically translational Killing vectors and the rigid positive energy
  theorem}.
\newblock {\em Journal of Mathematical Physics}, 37:1939--1961, Apr. 1996.

\bibitem{beig:omurchadha:1987AnPhy.174..463B}
R.~Beig and N.~{\'O}~Murchadha.
\newblock The {P}oincar\'e group as the symmetry group of canonical general
  relativity.
\newblock {\em Ann. Physics}, 174(2):463--498, 1987.

\bibitem{bergqvist:1999CMaPh.207..467B}
G.~{Bergqvist}.
\newblock {Positivity of General Superenergy Tensors}.
\newblock {\em Communications in Mathematical Physics}, 207:467--479, 1999.

\bibitem{bergqvist:etal:2003CQGra..20.2663B}
G.~{Bergqvist}, I.~{Eriksson}, and J.~M.~M. {Senovilla}.
\newblock {New electromagnetic conservation laws}.
\newblock {\em Classical and Quantum Gravity}, 20:2663--2668, July 2003.
\newblock \arxivref{gr-qc/0303036}.

\bibitem{bini:etal:2002PThPh.107..967B}
D.~{Bini}, C.~{Cherubini}, R.~T. {Jantzen}, and R.~{Ruffini}.
\newblock {Teukolsky Master Equation ---de Rham Wave Equation for the
  Gravitational and Electromagnetic Fields in Vacuum---}.
\newblock {\em Progress of Theoretical Physics}, 107:967--992, May 2002.
\newblock \arxivref{gr-qc/0203069}.

\bibitem{bini:etal:2003IJMPD..12.1363B}
D.~{Bini}, C.~{Cherubini}, R.~T. {Jantzen}, and R.~{Ruffini}.
\newblock {De {R}ham Wave Equation for Tensor Valued $p$-forms}.
\newblock {\em International Journal of Modern Physics D}, 12:1363--1384, 2003.

\bibitem{bray:MR1908823}
H.~L. Bray.
\newblock Proof of the {R}iemannian {P}enrose inequality using the positive
  mass theorem.
\newblock {\em J. Differential Geom.}, 59(2):177--267, 2001.

\bibitem{Buchdahl58}
H.~Buchdahl.
\newblock On the compatibility of relativistic wave equations for particles of
  higher spin in the presence of a gravitational field.
\newblock {\em Il Nuovo Cimento}, 10(1):96--103, 1958.

\bibitem{carter:1968PhRv..174.1559C}
B.~{Carter}.
\newblock {Global Structure of the Kerr Family of Gravitational Fields}.
\newblock {\em Physical Review}, 174:1559--1571, Oct. 1968.

\bibitem{chandrasekhar:MR1210321}
S.~Chandrasekhar.
\newblock {\em The mathematical theory of black holes}, volume~69 of {\em
  International Series of Monographs on Physics}.
\newblock The Clarendon Press, Oxford University Press, New York, 1992.
\newblock Revised reprint of the 1983 original, Oxford Science Publications.

\bibitem{christodoulou:MR885564}
D.~Christodoulou.
\newblock A mathematical theory of gravitational collapse.
\newblock {\em Comm. Math. Phys.}, 109(4):613--647, 1987.

\bibitem{1993gnsm.book.....C}
D.~{Christodoulou} and S.~{Klainerman}.
\newblock {\em {The global nonlinear stability of the Minkowski space}}.
\newblock 1993.

\bibitem{cohen:kegeles:1974PhRvD..10.1070C}
J.~M. {Cohen} and L.~S. {Kegeles}.
\newblock {Electromagnetic fields in curved spaces: A constructive procedure}.
\newblock {\em \prd}, 10:1070--1084, Aug. 1974.

\bibitem{1976IJTP...15..311C}
C.~D. {Collinson}.
\newblock {On the Relationship between Killing Tensors and Killing-Yano
  Tensors}.
\newblock {\em International Journal of Theoretical Physics}, 15:311--314, May
  1976.

\bibitem{dafermos:rodnianski:etal:2014arXiv1402.7034D}
M.~{Dafermos}, I.~{Rodnianski}, and Y.~{Shlapentokh-Rothman}.
\newblock {Decay for solutions of the wave equation on Kerr exterior spacetimes
  III: The full subextremal case $|a| < {M}$}.
\newblock Feb. 2014.
\newblock \arxivref{1402.7034}.

\bibitem{dai:wang:wei:MR2178660}
X.~Dai, X.~Wang, and G.~Wei.
\newblock On the stability of {R}iemannian manifold with parallel spinors.
\newblock {\em Invent. Math.}, 161(1):151--176, 2005.

\bibitem{edgar:etal:2009CQGra..26j5022E}
S.~B. {Edgar}, A.~G.-P. {G{\'o}mez-Lobo}, and J.~M.
  {Mart{\'{\i}}n-Garc{\'{\i}}a}.
\newblock {Petrov D vacuum spaces revisited: identities and invariant
  classification}.
\newblock {\em Classical and Quantum Gravity}, 26(10):105022, May 2009.
\newblock \arxivref{0812.1232}.

\bibitem{ferrando:saez:2007JMP....48j2504F}
J.~J. {Ferrando} and J.~A. {S{\'a}ez}.
\newblock {On the invariant symmetries of the $\mathcal{D}$-metrics}.
\newblock {\em Journal of Mathematical Physics}, 48(10):102504, Oct. 2007.
\newblock \arxivref{0706.3301}.

\bibitem{ferrando:saez:2009CQGra..26g5013F}
J.~J. {Ferrando} and J.~A. {S{\'a}ez}.
\newblock {An intrinsic characterization of the Kerr metric}.
\newblock {\em Classical and Quantum Gravity}, 26(7):075013, Apr. 2009.
\newblock \arxivref{0812.3310}.

\bibitem{finster:etal:MR2525736}
F.~Finster, N.~Kamran, J.~Smoller, and S.-T. Yau.
\newblock Linear waves in the {K}err geometry: a mathematical voyage to black
  hole physics.
\newblock {\em Bull. Amer. Math. Soc. (N.S.)}, 46(4):635--659, 2009.

\bibitem{finster:etal:energyextraction:MR2486663}
F.~Finster, N.~Kamran, J.~Smoller, and S.-T. Yau.
\newblock A rigorous treatment of energy extraction from a rotating black hole.
\newblock {\em Comm. Math. Phys.}, 287(3):829--847, 2009.

\bibitem{fushchich:nikitin:1992:JPhysA:MR1154854}
W.~I. Fushchich and A.~G. Nikitin.
\newblock The complete sets of conservation laws for the electromagnetic field.
\newblock {\em J. Phys. A}, 25(5):L231--L233, 1992.

\bibitem{2008PhRvD..77b4035G}
J.~R. {Gair}, C.~{Li}, and I.~{Mandel}.
\newblock {Observable properties of orbits in exact bumpy spacetimes}.
\newblock {\em \prd}, 77(2):024035, Jan. 2008.
\newblock \arxivref{0708.0628}.

\bibitem{Ger70spinstructII}
R.~Geroch.
\newblock Spinor structure of space-times in general relativity. {II}.
\newblock {\em Journal of Mathematical Physics}, 11(1):343--348, 1970.

\bibitem{GHP}
R.~{Geroch}, A.~{Held}, and R.~{Penrose}.
\newblock {A space-time calculus based on pairs of null directions}.
\newblock {\em Journal of Mathematical Physics}, 14:874--881, July 1973.

\bibitem{hacyan:1979PhLA...75...23H}
S.~{Hacyan}.
\newblock {Gravitational instantons in H-spaces}.
\newblock {\em Physics Letters A}, 75:23--24, Dec. 1979.

\bibitem{huisken:ilmanen:MR1916951}
G.~Huisken and T.~Ilmanen.
\newblock The inverse mean curvature flow and the {R}iemannian {P}enrose
  inequality.
\newblock {\em J. Differential Geom.}, 59(3):353--437, 2001.

\bibitem{kalnins:etal:1989JMP....30.2925K}
E.~G. {Kalnins}, W.~{Miller}, Jr., and G.~C. {Williams}.
\newblock {Teukolsky-Starobinsky identities for arbitrary spin}.
\newblock {\em Journal of Mathematical Physics}, 30:2925--2929, Dec. 1989.

\bibitem{Karlhede:1986}
A.~{Karlhede}.
\newblock {LETTER TO THE EDITOR: Classification of Euclidean metrics}.
\newblock {\em Classical and Quantum Gravity}, 3:L1--L4, Jan. 1986.

\bibitem{kerr:1963PhRvL..11..237K}
R.~P. {Kerr}.
\newblock {Gravitational Field of a Spinning Mass as an Example of
  Algebraically Special Metrics}.
\newblock {\em Physical Review Letters}, 11:237--238, Sept. 1963.

\bibitem{killing:1892}
W.~Killing.
\newblock {\"Uber} die grundlagen der geometrie.
\newblock {\em J. f\"ur den Reine und Angew. Mathematik}, 109:121--186, 1892.

\bibitem{kinnersley:1969JMP....10.1195K}
W.~{Kinnersley}.
\newblock {Type D Vacuum Metrics}.
\newblock {\em Journal of Mathematical Physics}, 10:1195--1203, July 1969.

\bibitem{lewandowski:1991CQGra...8L..11L}
J.~{Lewandowski}.
\newblock {Twistor equation in a curved spacetime}.
\newblock {\em Classical and Quantum Gravity}, 8:L11--L17, Jan. 1991.

\bibitem{2005CMaPh.256...43L}
H.~{Lindblad} and I.~{Rodnianski}.
\newblock {Global Existence for the Einstein Vacuum Equations in Wave
  Coordinates}.
\newblock {\em Communications in Mathematical Physics}, 256:43--110, May 2005.
\newblock \arxivref{math/0312479}.

\bibitem{lipkin:1964:MR0162484}
D.~M. Lipkin.
\newblock Existence of a new conservation law in electromagnetic theory.
\newblock {\em J. Mathematical Phys.}, 5:696--700, 1964.

\bibitem{lousto:whiting:2002PhRvD..66b4026L}
C.~O. {Lousto} and B.~F. {Whiting}.
\newblock {Reconstruction of black hole metric perturbations from the Weyl
  curvature}.
\newblock {\em \prd}, 66(2):024026, July 2002.

\bibitem{2010PhRvD..81l4005L}
G.~{Lukes-Gerakopoulos}, T.~A. {Apostolatos}, and G.~{Contopoulos}.
\newblock {Observable signature of a background deviating from the Kerr
  metric}.
\newblock {\em \prd}, 81(12):124005, June 2010.
\newblock \arxivref{1003.3120}.

\bibitem{mars:2000CQGra..17.3353M}
M.~{Mars}.
\newblock {Uniqueness properties of the Kerr metric}.
\newblock {\em Classical and Quantum Gravity}, 17:3353--3373, Aug. 2000.

\bibitem{xAct}
J.~M. Mart\'{\i}n-Garc\'{\i}a.
\newblock x{A}ct: {E}fficient tensor computer algebra for {M}athematica,
  2002-2014.
\newblock \href{http://www.xact.es}{http://www.xact.es}.

\bibitem{michel:radoux:silhan:2013arXiv1308.1046M}
J.-P. {Michel}, F.~{Radoux}, and J.~{{\v S}ilhan}.
\newblock {Second Order Symmetries of the Conformal Laplacian}.
\newblock {\em SIGMA}, 10:16, Feb. 2014.
\newblock \arxivref{1308.1046}.

\bibitem{MR0234136}
C.~S. Morawetz.
\newblock Time decay for the nonlinear {K}lein-{G}ordon equations.
\newblock {\em Proc. Roy. Soc. Ser. A}, 306:291--296, 1968.

\bibitem{FrolovNovikov}
I.~D. {Novikov} and V.~P. {Frolov}.
\newblock {\em {Physics of black holes}}.
\newblock (Fizika chernykh dyr, Moscow, Izdatel'stvo Nauka, 1986, 328 p)
  Dordrecht, Netherlands, Kluwer Academic Publishers, 1989, 351
  p.~Translation.~Previously cited in issue 19, p.~3128, Accession
  no.~A87-44677., 1989.

\bibitem{ONeill}
B.~O'Neill.
\newblock {\em The geometry of {K}err black holes}.
\newblock A K Peters Ltd., Wellesley, MA, 1995.

\bibitem{penrose:1973NYASA.224..125P}
R.~{Penrose}.
\newblock {Naked Singularities}.
\newblock In D.~J. {Hegyi}, editor, {\em Sixth Texas Symposium on Relativistic
  Astrophysics}, volume 224 of {\em Annals of the New York Academy of
  Sciences}, page 125, 1973.

\bibitem{Penrose:1986fk}
R.~Penrose and W.~Rindler.
\newblock {\em {Spinors and Space-time I {\&} II}}.
\newblock Cambridge Monographs on Mathematical Physics. Cambridge University
  Press, Cambridge, 1986.

\bibitem{perjes:lukacs_2005AIPC..767..306P}
Z.~{Perj{\'e}s} and {\'A}.~{Luk{\'a}cs}.
\newblock {Canonical Quantization and Black Hole Perturbations}.
\newblock In J.~{Lukierski} and D.~{Sorokin}, editors, {\em Fundamental
  Interactions and Twistor-Like Methods}, volume 767 of {\em American Institute
  of Physics Conference Series}, pages 306--315, Apr. 2005.

\bibitem{poisson:toolkit}
E.~Poisson.
\newblock {\em A relativist's toolkit}.
\newblock Cambridge University Press, Cambridge, 2004.
\newblock The mathematics of black-hole mechanics.

\bibitem{ryan:1974PhRvD..10.1736R}
M.~P. {Ryan}.
\newblock {Teukolsky equation and Penrose wave equation}.
\newblock {\em \prd}, 10:1736--1740, Sept. 1974.

\bibitem{senovilla:2000CQGra..17.2799S}
J.~M.~M. {Senovilla}.
\newblock {Super-energy tensors}.
\newblock {\em Classical and Quantum Gravity}, 17:2799--2841, July 2000.

\bibitem{stephani:etal:2009esef.book.....S}
H.~Stephani, D.~Kramer, M.~MacCallum, C.~Hoenselaers, and E.~Herlt.
\newblock {\em Exact solutions of {E}instein's field equations}.
\newblock Cambridge Monographs on Mathematical Physics. Cambridge University
  Press, Cambridge, second edition, 2003.

\bibitem{TataruTohaneanu}
D.~Tataru and M.~Tohaneanu.
\newblock A local energy estimate on {K}err black hole backgrounds.
\newblock {\em Int. Math. Res. Not.}, (2):248--292, 2011.

\bibitem{teukolsky:1972PhRvL..29.1114T}
S.~A. {Teukolsky}.
\newblock {Rotating Black Holes: Separable Wave Equations for Gravitational and
  Electromagnetic Perturbations}.
\newblock {\em Physical Review Letters}, 29:1114--1118, Oct. 1972.

\bibitem{teukolsky:1973}
S.~A. {Teukolsky}.
\newblock {Perturbations of a Rotating Black Hole. I. Fundamental Equations for
  Gravitational, Electromagnetic, and Neutrino-Field Perturbations}.
\newblock {\em Astrophysical J.}, 185:635--648, Oct. 1973.

\bibitem{walker:penrose:1970CMaPh..18..265W}
M.~{Walker} and R.~{Penrose}.
\newblock {On quadratic first integrals of the geodesic equations for type
  $\{$2,2$\}$ spacetimes}.
\newblock {\em Communications in Mathematical Physics}, 18:265--274, Dec. 1970.

\bibitem{M.Y.Wang:1989:MR1029845}
M.~Y. Wang.
\newblock Parallel spinors and parallel forms.
\newblock {\em Ann. Global Anal. Geom.}, 7(1):59--68, 1989.

\bibitem{M.Y.Wang:1991:MR1129331}
M.~Y. Wang.
\newblock Preserving parallel spinors under metric deformations.
\newblock {\em Indiana Univ. Math. J.}, 40(3):815--844, 1991.

\bibitem{znajek:1977MNRAS.179..457Z}
R.~L. {Znajek}.
\newblock {Black hole electrodynamics and the Carter tetrad}.
\newblock {\em \mnras}, 179:457--472, May 1977.

\end{thebibliography}

\end{document}